\documentclass[11pt]{article} 
\usepackage[T1]{fontenc}
\usepackage{lmodern}
\usepackage{geometry} 
\usepackage{graphicx}
\usepackage{amsmath,amssymb,amsthm,amsfonts,dsfont,mathtools}
\usepackage{relsize}
\usepackage[update,prepend]{epstopdf}
\usepackage{wrapfig}
\usepackage{frame,color}
\usepackage{hyperref}

\newtheorem{theorem}{Theorem}
\newtheorem{lemma}{Lemma}
\newtheorem{conjecture}{Conjecture}

\newtheorem{corollary}{Corollary}

\newcommand{\RNA}{\mathsf{RNA}}
\newcommand{\CG}{\mathsf{CG}}
\newcommand{\CNG}{\mathsf{CNG}}
\newcommand{\CLG}{\mathsf{CLG}}
\newcommand{\LG}{\mathsf{LG}}
\newcommand{\NNG}{\mathsf{NG}}
\newcommand{\dLCS}{\delta_{\operatorname{LCS}}}
\begin{document}
\date{}
\title{Hardness of RNA Folding Problem with Four Symbols\thanks{
A preliminary version of this paper was presented at the 27th Annual Symposium on Combinatorial Pattern Matching (CPM 2016)
Tel Aviv, Israel, June 27--29, 2016~\cite{chang:LIPIcs:2016:6089}.
Supported by NSF grants CCF-1217338, CNS-1318294,
and CCF-1514383. E-mail: cyijun@umich.edu.}}


\author{Yi-Jun Chang
\\ University of Michigan}


\maketitle


\begin{abstract}
An  RNA sequence is a string composed of four types of nucleotides,  $A, C, G$, and $U$. The goal of the RNA folding problem is to find a maximum cardinality set of crossing-free pairs of the form $\{A,U\}$ or $\{C,G\}$ in a given RNA sequence. The problem is central in bioinformatics and has received much attention over the years.  Abboud, Backurs, and Williams (FOCS 2015) demonstrated a conditional lower bound for a generalized version of the RNA folding problem based on a conjectured hardness of the $k$-clique problem. Their lower bound requires the RNA sequence to have at least 36 types of symbols, making the result not applicable to the RNA folding problem in real life (i.e., alphabet size 4). In this paper, we present an improved lower bound that works  for the alphabet size 4 case.

We also investigate the Dyck edit distance problem, which is a string problem closely related to RNA folding. We demonstrate a reduction from RNA folding to  Dyck edit distance with alphabet size 10. This leads to a much simpler proof of the conditional lower bound for Dyck edit distance problem given  by Abboud, Backurs, and Williams (FOCS 2015), and lowers the alphabet size requirement.
\end{abstract}


\setcounter{page}{1}

\section{Introduction \label{sec.intro}}

An {\em RNA sequence} is a string composed of four types of nucleotides,  namely $A, C, G$, and $U$. Given an RNA sequence, the goal of the {\em RNA folding} problem is to find a maximum cardinality set of crossing-free pairs of nucleotides, where all the pairs are either $\{A,U\}$ or $\{C,G\}$. The problem is central in bioinformatics, and it has found application in predicting the secondary structure of RNA molecules, which is of importance in molecular biology; see e.g.,~\cite{FG10} for more details.

It is well-known that the RNA folding problem can be solved in  $O(n^3)$ time via dynamic programming~\cite{DEKM98}. Due to the importance of the problem in practice, there has been a long line of research aiming at improving the runtime, practically or theoretically~\cite{A99,FG10,PTZZ11,PZTZ13,VGF13}. Based on log-shaving techniques such as the four-Russian method,
the time complexity of ${O}\left(\frac{n^3}{\log^2 n}\right)$ can be obtained~\cite{PZTZ13}.

Whether the RNA folding problem can be solved in truly sub-cubic time (i.e.,  ${O}(n^{3-\epsilon})$ time for some constant $\epsilon > 0$) had been a major open problem until recently. In 2016, Bringmann, Grandoni, Saha, and Williams~\cite{BFSW16} showed that the problem can be solved in randomized $O(n^{2.8244})$ time and deterministic $O(n^{2.8603})$ time via fast bounded-difference min-plus matrix multiplication. The RNA folding problem can be reduced to min-plus matrix multiplication~\cite{BFSW16}. The current state-of-the-art algorithm for min-plus matrix multiplication~\cite{R14} takes $O(n^3 / c^{\sqrt{\log n}})$ time, for some constant $c > 1$. Using this algorithm, one immediately obtains an RNA folding algorithm with the same time complexity. Bringmann, Grandoni, Saha, and Williams observed that the min-plus matrix multiplication instance resulting from the reduction has the ``bounded differences'' property, and they showed that bounded-difference min-plus matrix multiplication can be solved in truly sub-cubic time~\cite{BFSW16}.

The algorithm of~\cite{BFSW16} uses fast matrix multiplication, which does not perform very efficiently in practice.
It is still an open question whether there is a combinatorial, non-algebraic, truly sub-cubic time algorithm for RNA folding.

\subsection{Conditional Lower Bounds}
A popular way to show hardness of a problem is to demonstrate a lower bound conditioned on some widely accepted hypothesis.

\begin{conjecture} [Strong Exponential Time Hypothesis] \label{c-1}
There exist no $\epsilon > 0$ and $k_0 > 0$ such that $k$-SAT with $n$ variables can be solved in time $O(2^{(1-\epsilon)n})$ for all $k > k_0$.
\end{conjecture}

\begin{conjecture} [$k$-clique Conjecture] \label{c-2}
There exist no $\epsilon > 0$ and $k_0 > 0$ such that $k$-clique on graphs with $n$ nodes  can be solved in time $O\left(n^{(\omega - \epsilon) k/3}\right)$ for all $k > k_0$, where $\omega < 2.373$ is the matrix multiplication exponent.
\end{conjecture}

For instance, assuming Strong Exponential Time Hypothesis (SETH), the following bounds are unattainable for any $\epsilon > 0$:
(i) an ${O}(n^{k-\epsilon})$ algorithm for $k$-dominating set problem~\cite{PR10};
(ii) an ${O}(n^{2-\epsilon})$ algorithm for dynamic time warping, longest common subsequence, and edit distance~\cite{ABV15*,BI14,BK15};
(iii) an ${O}(m^{2-\epsilon})$ algorithm for ($3/2 - \epsilon$)-approximating the diameter of a graph with $m$ edges~\cite{RV13}.

As remarked in~\cite{ABV15}, it is straightforward to reduce the longest common subsequence (LCS) problem on binary strings to the RNA folding problem as follows. Given two binary strings $X, Y$,  let $\hat{X} \in {\{A,C\}}^{|X|}$ be a string defined as $\hat{X}[i] = A$ if $X[i] = 0$, and $\hat{X}[i] = C$ if $X[i] = 1$; similarly, let  $\hat{Y} \in {\{G,U\}}^{|Y|}$ be a string defined as $\hat{Y}[i] = U$ if $Y[i] = 0$, and $\hat{Y}[i] = G$ if $Y[i] = 1$. Then we have an 1-1 correspondence between   RNA foldings of $\hat{X} \circ \hat{Y}^R$ (i.e., concatenation of $\hat{X}$ and the reversal of $\hat{Y}$) and  common subsequences of $X$ and $Y$. It has been shown in~\cite{BK15} that there is no ${O}(n^{2-\epsilon})$-time algorithm for LCS on binary strings, assuming SETH. Thus, we immediately obtain the same conditional lower bound for RNA folding.

Abboud, Backurs, and Williams demonstrated  a higher conditional lower bound for a generalized version of the RNA folding problem (which coincides with the RNA folding problem when the alphabet size is 4) from the $k$-clique Conjecture~\cite{ABV15}.

\begin{theorem} [\cite{ABV15}] \label{thm-1}
If the generalized RNA folding problem on sequences of length $n$ with alphabet size 36  can be solved in $T(n)$ time, then $3k$-clique on graphs with $|V|=n$ can be solved in $T\left( O(n^{k + 2} \log n) \right)$ time.
\end{theorem}

Therefore, an ${O}(n^{\omega - \epsilon})$-time algorithm for the generalized RNA folding with alphabet size at least 36 will disprove the $k$-clique Conjecture, yielding  a breakthrough to the parameterized complexity of the $k$-clique problem.
The current state-of-the-art algorithm for $k$-clique takes ${{O}}\left(n^{\omega k/3}\right)$ time (when $k$ is a multiple of 3), and it requires the use of fast matrix multiplication~\cite{EG04}. For combinatorial, non-algebraic algorithms for $k$-clique, the current state-of-the-art upper bound is ${{O}}\left(\frac{n^k}{\log^k n}\right)$~\cite{V09}. Therefore, an ${O}(n^{3- \epsilon})$-time combinatorial algorithm for RNA folding would imply a breakthrough for combinatorial algorithms for $k$-clique.

\subsection{Our Contribution}
Due to its alphabet size requirement, Theorem~\ref{thm-1} is not applicable to the RNA folding problem in real life (i.e., alphabet size 4). It is unknown whether the generalized RNA folding for alphabet size $4$  admits a faster algorithm than the case for alphabet size $> 4$. There are plenty of string algorithms whose runtime depends on the alphabet size (e.g., string matching with mismatches~\cite{AL91} and jumbled indexing~\cite{ACLL14,CL15}). Note that when the alphabet size is 2, the generalized RNA folding problem can be trivially solved in linear time.
In this paper, we improve upon Theorem~\ref{thm-1} by showing the same conditional lower bound still for the case of alphabet size 4.
Note that we also get an ${O}(n)$ factor improvement inside $T(\cdot)$, though it does not affect the conditional lower bound.

\begin{theorem}  \label{thm-2}
If the RNA folding problem on sequences in ${\{A,C,G,U\}}^n$ can be solved in $T(n)$ time, then $3k$-clique on graphs with $|V|=n$ can be solved in $T\left( O(n^{k + 1} \log n)\right)$ time.
\end{theorem}

In what follows, we briefly overview the proof of Theorem~\ref{thm-2}.
At a high level, our reduction (from $3k$-clique to RNA folding) follows the approach in \cite{ABV15}. Given a graph $G$, enumerate all $k$-cliques, and each of them is encoded as some gadgets. All the gadgets are then put together to form an RNA sequence. The goal is to ensure that an optimal RNA folding corresponds to choosing three $k$-cliques that form a $3k$-clique, given that the underlying graph admits a $3k$-clique.

Intuitively, in order to force the gadgets to be matched in a desired manner in an optimal RNA folding, we have to build various ``walls'' that prevent undesired pairings.
The main challenge is that we have to achieve this goal using merely 4 types of symbols.
Our main tool is to use the technique ``alignment gadget'' developed in~\cite{BK15}, whose original purpose is to prove that longest common subsequence and other edit distance problems are SETH-hard even on binary strings. We apply this tool as a black box during the construction of the RNA sequence.



\paragraph{Dyck Edit Distance.}
The RNA folding problem can be alternatively defined as follows. Given a string $S$, delete the minimum number of letters in $S$ to transform it into another string $S'$ in the language defined by the grammar $\mathbf{S} \rightarrow \mathbf{SS},  A\mathbf{S}U, U\mathbf{S}A, C\mathbf{S}G, G\mathbf{S}C, \epsilon$ (empty string). The {\em Dyck edit distance problem}~\cite{S14,S15*}, which asks for the minimum number of edits to transform a given string to a well-balanced parentheses of $s$ different types, admits a similar formulation. Due to the similarity, the Dyck edit distance problem was shown to admit the same conditional lower bound as Theorem~\ref{thm-1}~\cite{ABV15}.
Their conditional lower bound requires the alphabet size to be at least $48$.
In this paper, we present  a simple reduction from RNA folding to Dyck edit distance.

\begin{theorem}\label{thm-3}
If Dyck edit distance problem on sequences of length $n$ with alphabet size 10 can be solved in $T(n)$ time, then the RNA folding problem on sequences in ${\{A,C,G,U\}}^n$ can be solved in ${O}(T(n))$ time.
\end{theorem}

Combining Theorem~\ref{thm-2} and Theorem~\ref{thm-3}, we obtain the following corollary.

\begin{corollary} \label{cor-1}
If the Dyck edit distance problem on sequences of length $n$ with alphabet size 10  can be solved in $T(n)$ time, then $3k$-clique on graphs with $|V|=n$ can be solved in $T\left( O(n^{k + 1} \log n)\right)$ time.
\end{corollary}

This improves upon the conditional lower bound in~\cite{ABV15} (reducing the alphabet size from $48$ to $10$), and it also simplifies the proof (the original proof in~\cite{ABV15} takes about 9 pages).

\section{Preliminaries \label{sec.prelim}}

Given a set of letters $\Sigma$, the set $\Sigma'$ is defined as $\{x' \ | \ x \in \Sigma\}$. It is required that $\Sigma \cap \Sigma' = \emptyset$, and $\forall x, y \in \Sigma,  (x \neq y) \rightarrow (x' \neq y')$. Therefore, $|\Sigma'| = |\Sigma|$ and $|\Sigma \cup \Sigma'| = 2|\Sigma|$.

For any string $X = (x_1, \ldots, x_k) \in \Sigma^k$, we write $p(X)$ to denote $(x_1', \ldots, x_k')$. The letter $p$ stands for the prime symbol. We denote the reversal of the sequence $X$ as $X^R$. The concatenation of two sequences $X, Y$ is denoted as $X \circ Y$, or simply $XY$. We write {\em substring} to denote a contiguous subsequence.
We say that two pairs of indices $(i_1,j_1)$, $(i_2,j_2)$, with $i_1 < j_1$ and $i_2 < j_2$, form a {\em crossing pair} if
$$\left( \{i_1, j_1\} \cap \{i_2,j_2\} \neq \emptyset \right) \vee
\left( i_1 < i_2 < j_1 < j_2 \right) \vee
\left( i_2 < i_1 < j_2 < j_1\right).$$

\paragraph{Generalized RNA Folding.}
Given a string $S \in (\Sigma \cup \Sigma')^n$,  an {\em RNA folding} of $S$ is a set $A \subseteq \{(i,j) \ | \ 1 \leq i < j \leq n \}$ meeting the following two conditions.
\begin{itemize}
\item $A$ does not contain any crossing pair.
\item For each $(i,j) \in A$, either $S[i] \in \Sigma$ and $S[j] = S[i]'$, or  $S[j] \in \Sigma$ and $S[i] = S[j]'$ is true.
\end{itemize}
The goal of the generalized RNA folding problem is to find a maximum cardinality RNA folding $A^\ast$.
We write $\RNA(S) = |A^\ast|$, where $A^\ast$ is any  maximum cardinality RNA folding of $S$.
Any RNA folding $A$ satisfying $|A| = \RNA(S)$ is said to be {\em optimal}.

In the paper we only focus on the generalized RNA folding problem with four types of letters, i.e. $\Sigma = \{0,1\}, \Sigma' = \{0',1'\}$, which coincides with the RNA folding problem for alphabet $\{A,C,G,U\}$.
With a slight abuse of notation, sometimes we  write $(S[i], S[j])$ to denote a pair $(i,j) \in A$. The notation $\{\cdot,\cdot\}$ is used to indicate an unordered pair.

\paragraph{Longest Common Subsequence.} Given $X \in \Sigma^n$ and $Y \in \Sigma^m$, we define $\dLCS (X,Y)$ as the minimum number of letters from $X$ and $Y$ needed to be deleted to make them identical.
That is, $\dLCS (X,Y) = n+m - 2k$, where $k$ is the length of the longest common subsequence (LCS) of $X$ and $Y$. Observe that $\RNA(X \circ p(Y^R)) = (n+m - \dLCS (X,Y))/2$ equals the length of LCS of $X$ and $Y$. Hence the LCS problem can be viewed as the RNA folding problem with some structural constraint on the RNA sequence.

\paragraph{Alignment Gadgets.}
In~\cite{BK15}, a SETH-based conditional lower bound for  LCS   with $|\Sigma| = 2$ was shown. A key technique in their approach is a function that transforms an instance of an alignment problem between two sets of sequences to an instance of the LCS problem, which is briefly described as follows.

Let $\mathbf{X} = (X_1, \ldots, X_n)$ and $\mathbf{Y} = (Y_1, \ldots, Y_m)$ be two  linearly ordered sets of sequences of alphabet $\Sigma$. We assume that $n \geq m$. An {\em alignment} is a set $A = \{(i_1, j_1), (i_2, j_2), \ldots$, $(i_{|A|}, j_{|A|})\}$ with $1 \leq i_1 < i_2 < \cdots < i_{|A|} \leq n$ and $1 \leq j_1 < j_2 < \cdots < j_{|A|} \leq m$. An alignment $A$ is called {\em structural} iff $|A| = m$ and $i_{m} = i_1 + m - 1$. That is, all sequences in $\mathbf{Y}$ are matched, and the matched positions in $\mathbf{X}$ are contiguous. The set of all alignments is denoted by $\mathcal{A}_{n,m}$, and  the set of all structural alignments is denoted by $\mathcal{S}_{n,m}$.
The {\em cost} of an alignment $A$ (with respect to $\mathbf{X}$ and $\mathbf{Y}$) is defined as
$$ \delta(A) = \sum_{(i,j) \in A} \dLCS(X_i, Y_j) + (m - |A|) \max_{i,j} \dLCS(X_i, Y_j). $$
That is, unaligned parts of $\mathbf{Y}$ are penalized by $\max_{i,j} \dLCS(X_i, Y_j)$.
Given a sequence $X$, the {\em type} of $X$ is defined as $(|X|, \sum_i X[i])$, where each letter is assumed to be a number. Note that if the alphabet is $\Sigma = \{0, 1\}$ (which is the case in this paper), then $\sum_i X[i]$ is the number of occurrences of $1$ in $X$.
The following lemma was proved in \cite{BK15}.\footnote{See Lemma~4.3 in the arXiv version (1504.01431v2) of~\cite{BK15}.}

\begin{lemma} [\cite{BK15}] \label{lem-1}
Let $\mathbf{X} = (X_1, \ldots, X_n)$ and $\mathbf{Y} = (Y_1, \ldots, Y_m)$ be two  linearly ordered sets of binary strings such that $n \geq m$. All $X_i$ are of type $\mathcal{T}_{X} = (\ell_{X}, s_{X})$, and all $Y_i$ are of type $\mathcal{T}_Y = (\ell_{Y}, s_{Y})$. There are two binary strings $S_X = \mathrm{GA}_{{X}}^{m, \mathcal{T}_{Y}}(X_1, \ldots, X_n)$, $S_Y = \mathrm{GA}_{{Y}}^{n, \mathcal{T}_{X}}(Y_1, \ldots, Y_m)$ and an integer $C$ meeting the following requirements.
\begin{itemize}
\item $\min_{A \in \mathcal{A}_{n,m}} \delta(A) \leq \dLCS(S_X, S_Y) - C \leq \min_{A \in \mathcal{S}_{n,m}} \delta(A)$.
\item the integer $C$ and the types of $S_X$ and $S_Y$ only depend on $n, m, \mathcal{T}_{X}, \mathcal{T}_{Y}$.
\item $S_X, S_Y$, and $C$ can be calculated in time ${O}((n+m)(\ell_X +\ell_Y))$; hence $|S_X|$ and $|S_Y|$ are both ${O}((n+m)(\ell_{X} +\ell_{Y}))$.
\end{itemize}
\end{lemma}

In Lemma~\ref{lem-1}, $\mathrm{GA}_{{X}}^{m, \mathcal{T}_{Y}}(X_1, \ldots, X_n)$ is a function of $X_1, \ldots, X_n$ parameterized by $m$ and $\mathcal{T}_{Y}$; and $\mathrm{GA}_{{Y}}^{n, \mathcal{T}_{X}}(Y_1, \ldots, Y_m)$ is a function of $Y_1, \ldots, Y_m$ parameterized by $n$ and $\mathcal{T}_{X}$. 

Intuitively, computing an optimal alignment (or an optimal structural alignment) of two sets of sequences is at least as hard as computing a longest common subsequence. Lemma~\ref{lem-1} gives a reduction from the computation of a number $s$ with $\min_{A \in \mathcal{A}_{n,m}} \delta(A) \leq s \leq \min_{A \in \mathcal{S}_{n,m}} \delta(A)$ (which can be regarded as an approximately optimal alignment) to a single LCS instance.

Whenever  Lemma~\ref{lem-1} is invoked in this paper, the two sets $\mathbf{X}$ and $\mathbf{Y}$ are chosen in such a way that $\min_{A \in \mathcal{A}_{n,m}} \delta(A) =  \min_{A \in \mathcal{S}_{n,m}} \delta(A)$.
In particular, we use Lemma~\ref{lem-1} as a black box to devise two encodings, the clique node gadget $\text{CNG}(t)$ and the clique list gadget $\text{CLG}(t)$, for any $k$-clique $t$ in a graph, in such a way that whether two $k$-cliques $t_1$ and $t_2$ form a $2k$-clique can be inferred from the value of $\dLCS (\text{CNG}(t_1), \text{CLG}(t_2))$.

\section{From Cliques to RNA Folding \label{sec.reduction}}

The goal of this section is to prove Theorem~\ref{thm-2}.
Let $G = (V,E)$ be a graph, and let $n = |V|$. We write $\mathcal{C}_k$ to denote the set of $k$-cliques in $G$. We fix $\Sigma = \{0, 1\}$. As in~\cite{ABV15}, we construct a sequence $S_G \in (\Sigma \cup \Sigma')^\ast$ such that we can decide whether $G$ has a $3k$-clique according to the value of ${\RNA}(S_G)$.
The building blocks in the construction of $S_G$ carry the same names as their analogues in~\cite{ABV15}, though they have different lower-level implementations.

The organization of this section is as follows. In Section~\ref{ss-1} we describe the two gadgets $\text{CNG}(t)$ and $\text{CLG}(t)$ for a $k$-clique $t$ based on the black box tool of Lemma~\ref{lem-1}. In Section~\ref{ss-2}, adapting the gadgets developed in Section~\ref{ss-1}, we present the definition of the binary sequence $S_G$.
In Section~\ref{ss-3}, we show that there exists an optimal RNA foldings of $S_G$ satisfying some good properties, and then we calculate the value of ${\RNA}(S_G)$ in Section~\ref{ss-4}.

\subsection{Testing $2k$-cliques via LCS \label{ss-1}}

We associate with each vertex $v \in V$  a unique ID in $\{0, 1, \ldots\ n-1\}$. Let $s_v$ be the binary encoding of the ID of $v$. Note that  $|s_v| = \lceil \log (n) \rceil$ for each vertex $v$. We define $\bar{v}$ as the binary string resulting from replacing each 0 in $s_v$ by 01 and replacing each 1 in $s_v$ by 10. Observe that (i) $\bar{v}$ is of type $\mathcal{T}_0 = (2\lceil \log (n) \rceil, \lceil \log (n) \rceil)$ for each $v \in V$, and (ii) $\dLCS (\bar{u}, \bar{v}) = 0$ if and only if $u = v$.


Let $v \in V$ be any vertex, and let $N(v) = \{{u_1}, {u_2},\ldots, {u_{|N(v)|}}\}$ be the set of neighbors of $v$. The {\em list gadget} $\LG(v)$ and the {\em node gadget} $\NNG(v)$ for the vertex $v$ are defined as follows.
\begin{align*}
\LG(v) &= \mathrm{GA}_{{X}}^{1, \mathcal{T}_0}\left(\bar{u}_1, \bar{u}_2, \ldots, \bar{u}_{|N(v)|}, 1^{\lceil \log (n) \rceil}0^{\lceil \log (n) \rceil}, \ldots, 1^{\lceil \log (n) \rceil}0^{\lceil \log (n) \rceil} \right),\\
    &\text{where the number of occurrences of $1^{\lceil \log (n) \rceil}0^{\lceil \log (n) \rceil}$ is $n - |N(v)|$.}\\
\NNG(v) &= \mathrm{GA}_{{Y}}^{n, \mathcal{T}_0}(\bar{v}).
\end{align*}

\begin{lemma}\label{lem-2}
There is a number $c_0$, depending only on $n$, such that for any two vertices $v_1, v_2 \in V$, we have $\{v_1, v_2\} \in E$  if and only if $\dLCS(\LG(v_1), \NNG(v_2)) = c_0 = \min_{v_1', v_2' \in V} \dLCS (\LG(v_1'), \NNG(v_2'))$.
\end{lemma}
\begin{proof}
Let  $v_1, v_2 \in V$.
Let $N(v_1) = \{{u_1, u_2, \ldots, u_{|N(v_1)|}}\}$. Define the two sequences of binary strings $\mathbf{X}$ and $\mathbf{Y}$ as follows.
\begin{align*}
\mathbf{X} &= \left(\bar{u}_1, \bar{u}_2, \ldots, \bar{u}_{|N(v_1)|}, 1^{\lceil \log (n) \rceil}0^{\lceil \log (n) \rceil}, \ldots, 1^{\lceil \log (n) \rceil}0^{\lceil \log (n) \rceil}\right),\\
    &\text{where  the number of occurrences of $1^{\lceil \log (n) \rceil}0^{\lceil \log (n) \rceil}$ is $n - |N(v_1)|$.}\\
\mathbf{Y} &= (\bar{v}_2).
\end{align*}
Note that $|\mathbf{X}| = n$ and $|\mathbf{Y}|=1$; we have $\LG(v) = \mathrm{GA}_{{X}}^{1, \mathcal{T}_0}(\mathbf{X})$ and
$\NNG(v) = \mathrm{GA}_{{Y}}^{n, \mathcal{T}_0}(\mathbf{Y})$.

In view of Lemma~\ref{lem-1}, we have $\min_{A \in \mathcal{A}_{n,1}} \delta(A) \leq \dLCS(\LG(v_1)$, $\NNG(v_2)) - C \leq \min_{A \in \mathcal{S}_{n,1}} \delta(A)$, for some number $C$ whose value depends on $|\mathbf{X}|, |\mathbf{Y}|$, and $\mathcal{T}_0$. As these parameters depend solely on $n$, the number $C$ also depends only on $n$ (i.e., the choice of the two vertices $v_1$ and $v_2$ does not affect $C$).
We claim that setting $c_0 = C$ suffices to prove the lemma.

Since $|\mathbf{Y}|=1$, any non-empty alignment between $\mathbf{X}$ and $\mathbf{Y}$ is structural. This implies that $\dLCS(\LG(v_1)$, $\NNG(v_2)) - C = \min_{A \in \mathcal{A}_{n,1}} \delta(A) = \min_{A \in \mathcal{S}_{n,1}} \delta(A)$.

For the case $\{v_1, v_2\} \in E$, since $\bar{v}_2$ is contained in $\mathbf{X}$, clearly $\min_{A \in \mathcal{S}_{n,m}} \delta(A) = 0$. For the case $\{v_1, v_2\} \not\in E$, $\bar{v}_2$ does not appear in $\mathbf{X}$, so $\min_{A \in \mathcal{S}_{n,m}} \delta(A) > 0$. Note that $1^{\lceil \log (n) \rceil}0^{\lceil \log (n) \rceil} \neq \bar{v}$, for each $v \in V$.
As a result, $\{v_1, v_2\} \in E$  if and only if $\dLCS(\LG(v_1)$, $\NNG(v_2)) = C = \min_{v_1', v_2' \in V} \dLCS (\LG(v_1'), \NNG(v_2'))$, as desired.
\end{proof}

In view of Lemma~\ref{lem-1}, the type of list gadgets and the type of node gadgets depends only on $n$; that is, they are independent of the underlying vertex $v$. Let $\mathcal{T}_X$ be the type of the list gadgets, and let $\mathcal{T}_Y$ be the type of the node gadgets.  For each $k$-clique $t = \{u_1, u_2, \ldots, u_k\}$, define the {\em clique node gadget} $\CNG(t)$ and the {\em clique list gadget} $\CLG(t)$ as follows.
\begin{align*}
\CLG(t)
&= \mathrm{GA}_{{X}}^{k^2, \mathcal{T}_Y}\left(\LG(u_1), \ldots, \LG(u_1), \LG(u_2), \ldots, \LG(u_2), \ldots, \LG(u_k), \ldots, \LG(u_k)\right),\\
&\text{where the number of occurrences of each $\LG(u_i)$ is $k$.}\\
\CNG(t) &=
\mathrm{GA}_{{Y}}^{k^2, \mathcal{T}_X}\big(
\NNG(u_1), \NNG(u_2), \ldots, \NNG(u_k),\\
&\hspace{2.1cm}\NNG(u_1), \NNG(u_2), \ldots, \NNG(u_k),\\
&\hspace{2.1cm}\ldots,\\
&\hspace{2.1cm}\NNG(u_1), \NNG(u_2), \ldots, \NNG(u_k)\big),\\
&\text{where the number of occurrences of each $\NNG(u_1), \NNG(u_2), \ldots, \NNG(u_k)$ is $k$.}
\end{align*}

\begin{lemma}\label{lem-3}
There is a number $c_1$, depending only on $n$ and $k$, such that for any two $k$-cliques $t_1, t_2 \in \mathcal{C}_k$, the set of vertices $t_1 \cup t_2$ form a $2k$-clique if and only if $\dLCS (\CLG(t_1), \CNG(t_2)) = c_1 = \min_{t_1', t_2' \in \mathcal{C}_k} \dLCS (\CLG(t_1'), \CNG(t_2'))$.
\end{lemma}
\begin{proof}
Let $t_1 = \{u_1, u_2, \ldots, u_k\}$, and let $t_2 = \{v_1, v_2, \ldots, v_k\}$ be two $k$-cliques.
Define the two sequences of binary strings $\mathbf{X}$ and $\mathbf{Y}$ as follows.
\begin{align*}
\mathbf{X}
&= \left(\LG(u_1), \ldots, \LG(u_1), \LG(u_2), \ldots, \LG(u_2), \ldots, \LG(u_k), \ldots, \LG(u_k)\right),\\
&\text{where the number of occurrences of each $\LG(u_i)$ is $k$.}\\
\mathbf{Y} &=
\big(
\NNG(v_1), \NNG(v_2), \ldots, \NNG(v_k),\\
&\hspace{0.7cm}\NNG(v_1), \NNG(v_2), \ldots, \NNG(v_k),\\
&\hspace{0.7cm}\ldots,\\
&\hspace{0.7cm}\NNG(v_1), \NNG(v_2), \ldots, \NNG(v_k)\big),\\
&\text{where the number of occurrences of each $\NNG(v_1), \NNG(v_2), \ldots, \NNG(v_k)$ is $k$.}
\end{align*}
Note that $|\mathbf{X}| = |\mathbf{Y}| = k^2$; we have
$\CLG(t) = \mathrm{GA}_{{X}}^{k^2, \mathcal{T}_Y}(\mathbf{X})$ and
$\CNG(t) = \mathrm{GA}_{{Y}}^{k^2, \mathcal{T}_X}(\mathbf{Y})$.
In view of Lemma~\ref{lem-2}, $\min_{{w_1}, {w_2} \in V} \dLCS (\LG({w_1}), \NNG({w_2})) = c_0$, and so $\min_{A \in \mathcal{A}_{k^2,k^2}} \delta(A) \geq k^2 c_0$.

In view of Lemma~\ref{lem-1}, $\min_{A \in \mathcal{A}_{k^2,k^2}} \delta(A) \leq \dLCS(\CLG(t_1)$, $\CNG(t_2)) - C \leq \min_{A \in \mathcal{S}_{k^2,k^2}} \delta(A)$, for some number $C$ whose value depends on $|\mathbf{X}|$, $|\mathbf{Y}|$, $\mathcal{T}_X$, and $\mathcal{T}_Y$. As these parameters depend solely on $n, k$, the number $C$ only depends on $n, k$ (i.e., the choice of the two $k$-cliques $t_1$ and $t_2$ does not affect $C$). We claim that setting $c_1 = C + k^2 c_0$ suffices to prove the lemma.

Consider the case $t_1 \cup t_2$ form a $2k$-clique.
That is, each $u_i \in t_1$ is adjacent to each $v_j \in t_2$.
 Thus, by Lemma~\ref{lem-2}, we have $\dLCS(X_i, Y_j) = c_0$, for all $i, j$.
Recall that $X_i$ denotes the $i$th string in $\mathbf{X}$, and $Y_j$ denotes the $j$th string in $\mathbf{Y}$.
The structural alignment $A^\ast = \{(i, i) \ | \ i \in \{1, 2, \ldots, k^2\}\} \in \mathcal{S}_{k^2,k^2}$ achieves the minimum possible cost $k^2 c_0$.
Thus, for the case $t_1 \cup t_2$ form a $2k$-clique, we have
$$\dLCS(\CLG(t_1), \CNG(t_2)) - C  =
\min_{A \in \mathcal{A}_{k^2,k^2}} \delta(A) =
\min_{A \in \mathcal{S}_{k^2,k^2}} \delta(A) =
  k^2 c_0,$$ and so
$\dLCS(\CLG(t_1), \CNG(t_2)) = C + k^2 c_0 = c_1$, as desired.

Next, consider the case $t_1 \cup t_2$ does not form a $2k$-clique.
That is, there exist two vertices $u_{i'} \in t_1$  and $v_{j'} \in t_2$ that are not adjacent.
by Lemma~\ref{lem-2}, we have $\dLCS(\LG(u_{i'}), \NNG(u_{j'})) > c_0$.
We claim that $\min_{A \in \mathcal{A}_{k^2,k^2}} \delta(A) > k^2 c_0$.
Suppose that there exists an alignment $A' \in \mathcal{A}_{k^2,k^2}$ such that  $\delta(A') = k^2 c_0$.
Then all $k^2$ strings in $\mathbf{Y}$ must be aligned, as each unaligned string in $\mathbf{Y}$ contributes a cost that is higher than $c_0$.
Thus, we must have $A' = \{(i, i) \ | \ i \in \{1, 2, \ldots, k^2\}\}$.
In order to have $\delta(A') = k^2 c_0$, we must have $\dLCS(X_i, Y_i) = c_0$, for all $i \in \{1, 2, \ldots, k^2\}$.
However,  $X_{j'+k(i'-1)} =  \LG(u_{i'})$ and $Y_{j'+k(i'-1)} =  \NNG(v_{j'})$, and so $\dLCS(X_{j'+k(i'-1)}, Y_{j'+k(i'-1)}) > c_0$, a contradiction.

Since $\min_{A \in \mathcal{A}_{k^2,k^2}} \delta(A) > k^2 c_0$ for the case $t_1 \cup t_2$ is not a $2k$-clique, we have
$$\dLCS (\CLG({t_1}), \CNG({t_2})) - C  \geq \min_{A \in \mathcal{A}_{k^2,k^2}} \delta(A) > k^2 c_0,$$ and so
$\dLCS (\CLG({t_1}), \CNG({t_2})) > C + k^2 c_0 = c_1$, as desired.
\end{proof}

The following lemma is a simple consequence of Lemma~\ref{lem-1}.

\begin{lemma}\label{lem-length}
There exist four integers $\ell_{\CNG, 0}$, $\ell_{\CNG, 1}$, $\ell_{\CLG, 0}$, and $\ell_{\CLG, 1} \in {O}(k^2 n \log n)$ such that for any $t \in \mathcal{C}_k$, and for any $b \in \{0,1\}$, we have (i) $\ell_{\CNG, b}$ is the number of occurrences of $b$ in $\CNG(t)$, and (ii)  $\ell_{\CLG, b}$ is the number of occurrences of $b$ in $\CLG(t)$.
\end{lemma}

\begin{proof}
As a consequence of Lemma~\ref{lem-1}, all $\CNG(t)$ have the same type, and all $\CLG(t)$ have the same type. Therefore, the existence of these four integers is guaranteed. We show that these numbers are ${O}(k^2 n \log n)$.
In view of Lemma~\ref{lem-1}, for each $v \in V$, both $\LG(v)$ and  $\NNG(v)$ have length at most $(n+1) \cdot (2 \lceil \log n \rceil +　2 \lceil \log n \rceil ) = {O} (n \log n)$. Applying Lemma~\ref{lem-1} again, the length of both  $\CNG(t)$ and $\CLG(t)$  for each $t \in \mathcal{C}_k$ is $(k^2 + k^2)({O} (n \log n) + {O} (n \log n)) = {O}(k^2 n \log n)$.
\end{proof}


\subsection{The RNA sequence $S_G$  \label{ss-2}}

Based on the parameters in Lemma~\ref{lem-length}, we define $\ell_0 = \ell_{\CNG, 0} + \ell_{\CNG, 1} +\ell_{\CLG, 0} +\ell_{\CLG, 1} = {O}(k^2 n \log n)$; for each $i \in \{1,2,3\}$, we set $\ell_i = 100 \ell_{i-1}$; and $\ell_4 = 100 |\mathcal{C}_k| \ell_3
= {n \choose k} {O}(k^2 n \log n)
= {O}(n^{k+1} \log n / (k-2)!)$.
The RNA sequence $S_G$ is then defined as follows.
$$S_G = 0^{\ell_4} \left[{0'}^{\ell_3} \underset{{t \in \mathcal{C}_k}}{\bigcirc}  \left( \CG_\alpha(t){0'}^{\ell_3} \right) \right]
0^{\ell_4} \left[{0'}^{\ell_3} \underset{{t \in \mathcal{C}_k}}{\bigcirc}  \left(  \CG_\beta(t) {0'}^{\ell_3} \right)\right]
0^{\ell_4} \left[{0'}^{\ell_3} \underset{{t \in \mathcal{C}_k}}{\bigcirc}  \left(  \CG_\gamma(t) {0'}^{\ell_3} \right)\right],$$
where
\begin{align*}
\CG_\alpha(t) &= {1'}^{2 \ell_2} p({\CLG(t)}^R) {0'}^{\ell_1}  {1}^{\ell_2} {0}^{\ell_1} \CNG(t) {1}^{\ell_2},\\
\CG_\beta(t) &= {1'}^{\ell_2} p({\CLG(t)}^R) {0'}^{\ell_1}  {1'}^{2\ell_2} {0'}^{\ell_1} p(\CNG(t)) {1'}^{\ell_2},\\
\CG_\gamma(t) &= {1}^{\ell_2} {\CLG(t)}^R {0}^{\ell_1}  {1}^{\ell_2} {0}^{\ell_1} \CNG(t) {1}^{2\ell_2}.
\end{align*}
For each $t \in \mathcal{C}_k$, and for each $x \in \{\alpha, \beta, \gamma\}$, the string $\CG_x(t)$ is called a {\em clique gadget}.
Note that $\CG_\alpha(t) \in {(\Sigma \cup \Sigma')}^\ast$, $\CG_\beta(t) \in {\Sigma'}^\ast$, and $\CG_\gamma(t) \in {\Sigma}^\ast$.
The length of this RNA sequence is $|S_G| = {O}(|\mathcal{C}_k| \ell_0) = {O}(n^{k+1} \log n / (k-2)!)$.
Before proceeding further, we present a simple lower bound on ${\RNA}(S_G)$ by constructing an RNA folding of $S_G$ as follows.

\begin{description}
\item[Step~1: Matching the Letters in ${0'}^{\ell_3}$.]
Given some pairings between the letters in  ${0'}^{\ell_3}$ and the letters in  ${0}^{\ell_4}$, a clique gadget $C$ is said to be {\em blocked} if all letters within $C$ can only be paired up with the letters within the same clique gadget or the letters in ${0}^{\ell_4}$.
In particular, a clique gadget that is blocked is unable to participate in the RNA folding with other clique gadgets.

We link all $0'$ in all ${0'}^{\ell_3}$ to some $0$ in some ${0}^{\ell_4}$ in such a way that for each $x \in \{\alpha, \beta, \gamma\}$, there is exactly one clique gadget $\CG_\alpha(t_x)$ that is not blocked, among all clique gadgets $\{\CG_x(t) \ | \ t \in \mathcal{C}_k\}$.
The three clique gadgets $\CG_\alpha(t_\alpha)$, $\CG_\beta(t_\beta)$, and $\CG_\gamma(t_\gamma)$ that are not blocked are called the {\em selected} clique gadgets. See Fig.~\ref{fig-1}.
This step makes $3(|\mathcal{C}_k|+1) \ell_3$ pairs.

\item[Step~2: Matching the Letters in a Blocked Clique Gadget.]
We pair up the letters in each blocked clique gadget as follows.
For each blocked $\CG_\alpha(t)$, we pair up as many $\{1', 1\}$ pairs as possible within the clique gadget;
this gives us  $2 \ell_2 + \min(\ell_{\CLG, 1}, \ell_{\CNG, 1})$ pairs.
For each blocked $\CG_\beta(t)$, we pair up all $0'$ to some $0$ in ${0}^{\ell_4}$;
this gives us  $2 \ell_1 + \ell_{\CLG, 0} + \ell_{\CNG, 0}$ pairs.
For each blocked $\CG_\gamma(t)$, no pairing can be made. See Fig.~\ref{fig-2}.
In this step,  $(|\mathcal{C}_k|-1)( 2 \ell_1 + 2 \ell_2
 + \min (\ell_{\CLG, 1}, \ell_{\CNG, 1})
 + \ell_{\CLG, 0} + \ell_{\CNG, 0} )$ pairs are produced.

\item[Step~3: Matching the Letters in the Three Selected Clique Gadget.]
For the three clique gadgets  $\CG_\alpha(t_\alpha)$, $\CG_\beta(t_\beta)$, and $\CG_\gamma(t_\gamma)$ that are not blocked,
We pair up the letters in these clique gadgets in such a way that
\begin{itemize}
\item all letters in ${1'}^{2 \ell_2}$, ${1}^{2 \ell_2}$, ${1'}^{\ell_2}$, ${1}^{\ell_2}$, ${0'}^{\ell_1}$, and ${0}^{\ell_1}$ are matched, and
    \item for each $(x,y) \in \{(\alpha,\beta),(\alpha,\gamma),(\beta,\gamma)\}$,  $\frac{1}{2} (  \ell_0  - \dLCS (\CLG(t_x), \CNG(t_y)) )$ number of pairs are made between the two gadgets $\CLG(t_x)$ and $\CNG(t_y)$.
\end{itemize}
See Fig.~\ref{fig-3}. Recall that $\frac{1}{2} (  \ell_0  - \dLCS (\CLG(t_x), \CNG(t_y)) )$ is the length of the LCS between $\CLG(t_x)$ and $\CNG(t_y)$.
The total number of pairs made in this step is
\begin{align*}
&6 \ell_2 + 3 \ell_1
+   \frac{1}{2} \left(  \ell_0  - \dLCS (\CLG(t_\alpha), \CNG(t_\beta)) \right) \\
&+    \frac{1}{2} \left(  \ell_0  - \dLCS (\CLG(t_\alpha), \CNG(t_\gamma)) \right)
+    \frac{1}{2}  \left( \ell_0  - \dLCS (\CLG(t_\beta), \CNG(t_\gamma)) \right).
    \end{align*}
\end{description}

\begin{figure}
\begin{center}
    \includegraphics[width=1\textwidth]{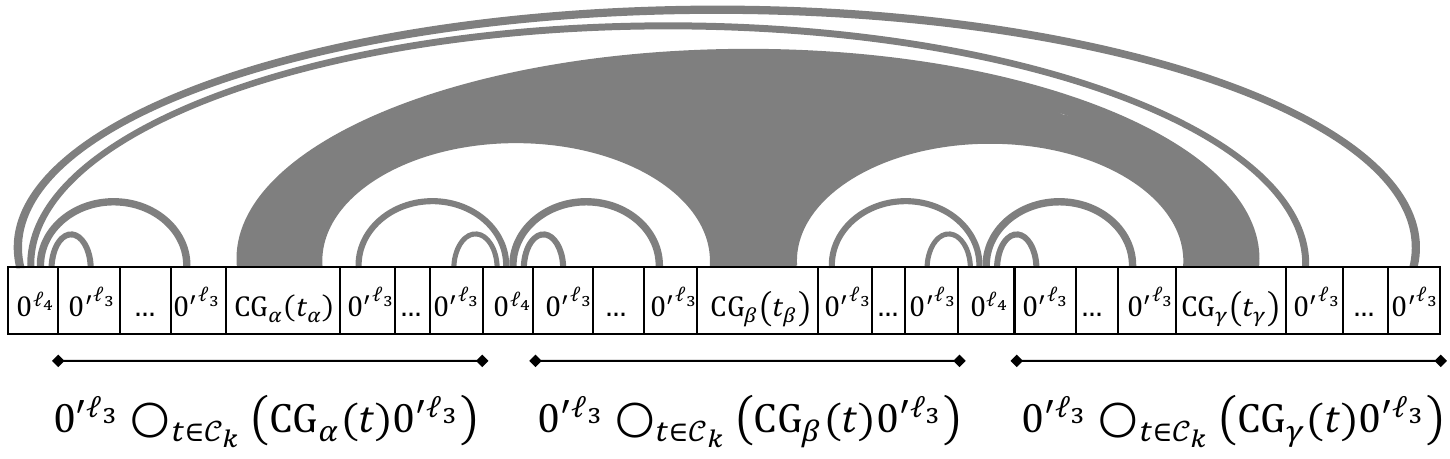}
    \end{center}
    \caption{The three selected clique gadgets and the matchings between ${0'}^{\ell_3}$ and ${0}^{\ell_4}$.
    } \label{fig-1}
\end{figure}

\begin{figure}
\begin{center}
    \includegraphics[width=1\textwidth]{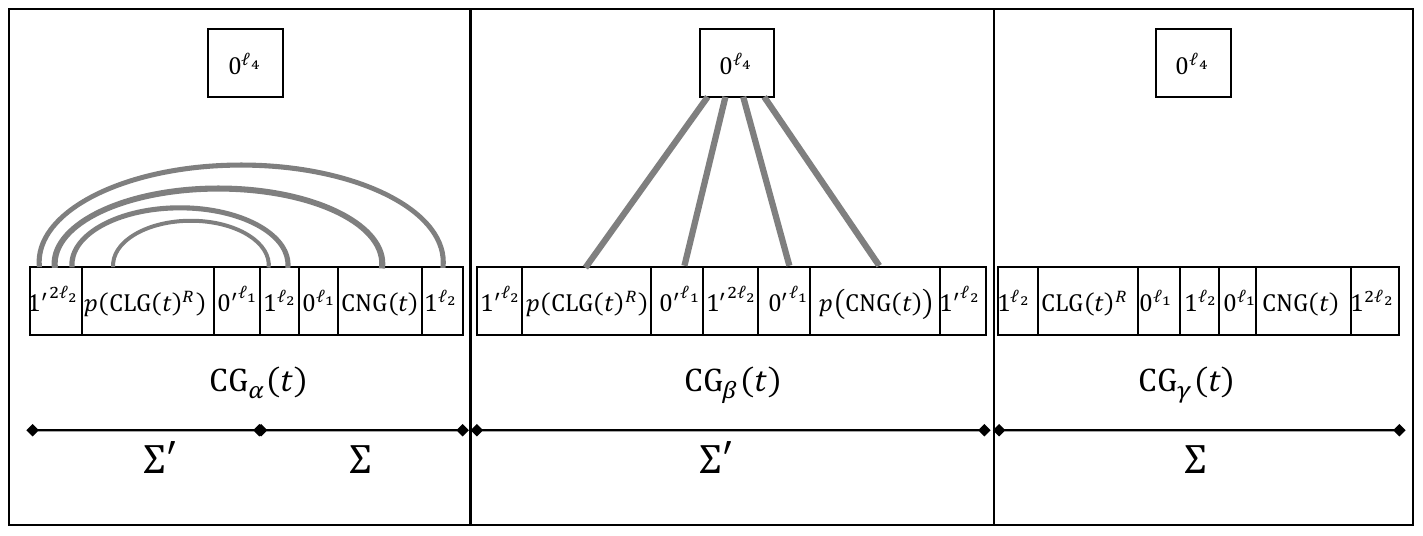}
    \end{center}
    \caption{The matchings between a blocked clique gadget and ${0}^{\ell_4}$.
    } \label{fig-2}
\end{figure}

\begin{figure}
\begin{center}
    \includegraphics[width=1\textwidth]{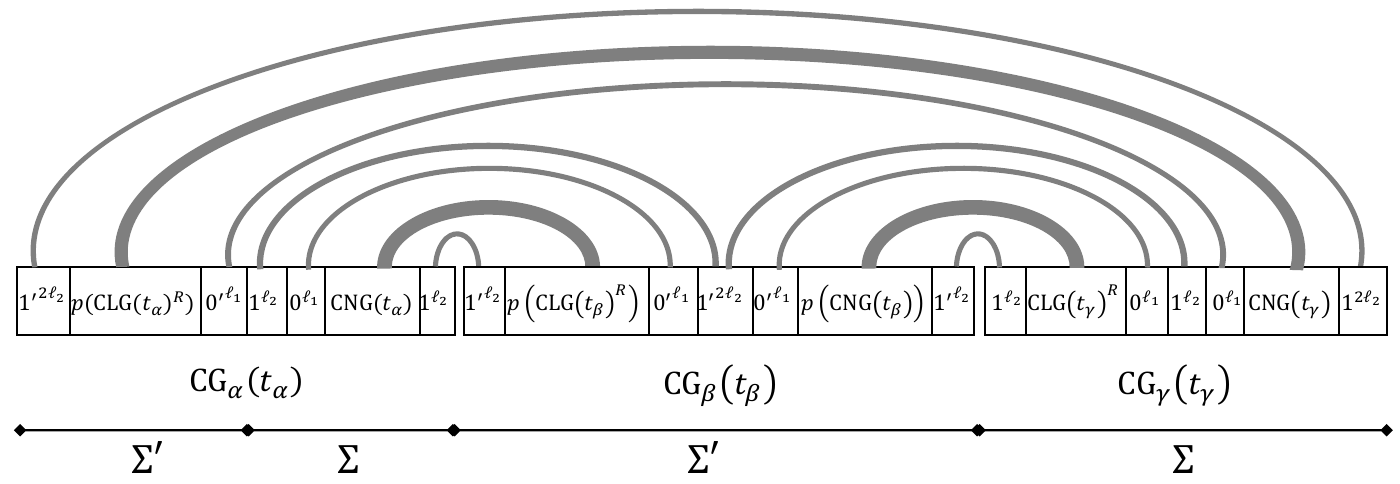}
    \end{center}
    \caption{The matchings within the three selected clique gadgets.
    } \label{fig-3}
\end{figure}

In view of the above discussion, we define the following two numbers.
\begin{align*}
m_1 &= 3(|\mathcal{C}_k| +1) \ell_3 + (|\mathcal{C}_k|-1) ( 2 \ell_1 + 2 \ell_2
 + \min (\ell_{\CLG, 1}, \ell_{\CNG, 1})
 + \ell_{\CLG, 0} + \ell_{\CNG, 0} ), \\
m_2 &= 6 \ell_2 + 3 \ell_1 + \frac{3}{2} \ell_0 - \min_{t_\alpha, t_\beta, t_\gamma \in \mathcal{C}_k} Q(t_\alpha, t_\beta, t_\gamma), \text{ where $Q(t_\alpha, t_\beta, t_\gamma)$ is defined as}\\
&\frac{1}{2}\left(
 \dLCS (\CLG(t_\alpha), \CNG(t_\beta)) +
 \dLCS (\CLG(t_\alpha), \CNG(t_\gamma)) +
 \dLCS (\CLG(t_\beta), \CNG(t_\gamma))
 \right).
\end{align*}

The RNA folding given in the above construction has cardinality $m_1 + 6 \ell_2 + 3 \ell_1 + \frac{3}{2} \ell_0 -  Q(t_\alpha, t_\beta, t_\gamma)$, and so $m_1 + m_2$ is a lower bound of ${\RNA}(S_G)$.

\begin{lemma}\label{lem-5}
${\RNA}(S_G) \geq m_1 + m_2$.
\end{lemma}

We will ultimately show that ${\RNA}(S_G) = m_1 + m_2$.
Due to Lemma~\ref{lem-3}, if $t_\alpha \cup t_\beta \cup t_\gamma$ form a $3k$-clique, then $Q(t_\alpha, t_\beta, t_\gamma) = 3 c_1 / 2$; otherwise $Q(t_\alpha, t_\beta, t_\gamma) > 3 c_1 / 2$. Therefore, the number ${\RNA}(S_G) = m_1 + m_2$ offers sufficient information to decide whether $G$ has a $3k$-clique.
The following auxiliary lemma will be useful in subsequent discussion.

\begin{lemma}\label{lem-4}
The following statements are true for any two cliques $t,t' \in \mathcal{C}_k$.
\begin{enumerate}
\item $\RNA( 0^{\ell_4} \CG_\alpha(t) ) = 2 \ell_2 + \min (\ell_{\CLG, 1}, \ell_{\CNG, 1})$
\item $\RNA( 0^{\ell_4} \CG_\beta(t) ) = 2 \ell_1 + \ell_{\CLG, 0} + \ell_{\CNG, 0}$
\item $\RNA( 0^{\ell_4} \CG_\gamma(t) ) = 0$
\item $\RNA( 0^{\ell_4} \CG_\alpha(t) 0^{\ell_4} \CG_\beta(t') ) \leq 3.1 \ell_1 + 2 \ell_2$
\item $\RNA( 0^{\ell_4} \CG_\alpha(t) 0^{\ell_4} \CG_\gamma(t') ) \leq 1.1 \ell_1 + 2 \ell_2$
\item $\RNA( 0^{\ell_4} \CG_\beta(t) 0^{\ell_4} \CG_\gamma(t') ) \leq 1.1 \ell_1 + 4 \ell_2$
\end{enumerate}
\end{lemma}
\begin{proof}
The value of $\RNA(\cdot)$  for each of the six strings are calculated as follows.
\begin{description}
\item[(1)~$\RNA( 0^{\ell_4} \CG_\alpha(t) ) = 2 \ell_2 + \min (\ell_{\CLG, 1}, \ell_{\CNG, 1})$:] Pairing up as many $1$ to  $1'$ yields a matching of size $m = 2 \ell_2 + \min (\ell_{\CLG, 1}, \ell_{\CNG, 1})$. To see that it is optimal, it suffices to show that both $(0',0)$ and $(0,0')$ cannot appear in an optimal RNA folding.
    \begin{itemize}
    \item If the RNA folding contains $(0,0')$, then none of $1'$ can participate in the RNA folding. As the total number of $0'$ is $\ell_1 + \ell_{\CLG, 0}$, the size of RNA folding is at most $\ell_1 + \ell_{\CLG, 0} < m$.
    \item If the RNA folding contains $(0',0)$, then at most $\ell_{\CLG, 1}$ number of letters within the middle ${1}^{\ell_2}$ (the one between ${0'}^{\ell_1}$ and ${0}^{\ell_1}$) can participate in the RNA folding. It implies that the number of $(1',1)$ pairs  in the RNA folding is at most $\ell_{\CLG, 1} + \ell_2$. Hence the size of the RNA folding can be upper bounded by $(\ell_1 + \ell_{\CLG, 0}) + (\ell_{\CLG, 1} + \ell_2) < m$.
    \end{itemize}
\item[(2)~$\RNA( 0^{\ell_4} \CG_\beta(t) ) = 2 \ell_1 + \ell_{\CLG, 0} + \ell_{\CNG, 0}$:] Since there is no $1$, the equation follows from the fact that there are $ 2 \ell_1 + \ell_{\CLG, 0} + \ell_{\CNG, 0}$ occurrences of $0'$,  all of which can be matched to some $0$ without crossing.
\item[(3)~$\RNA( 0^{\ell_4} \CG_\gamma(t) ) = 0$:] It is impossible to produce any pair since there are no $0'$ and $1'$ in the string.
\item[(4)~$\RNA( 0^{\ell_4} \CG_\alpha(t) 0^{\ell_4} \CG_\beta(t') ) \leq 3.1 \ell_1 + 2 \ell_2$:] The value of $\RNA(\cdot)$ can be upper bounded by the number of $1$ and $0'$, which is $(2 \ell_2 + \ell_{\CNG, 1}) + (3 \ell_1 +　2\ell_{\CLG, 0} +　\ell_{\CNG, 0}) \leq 3.1 \ell_1 + 2 \ell_2$.
\item[(5)~$\RNA( 0^{\ell_4} \CG_\alpha(t) 0^{\ell_4} \CG_\gamma(t') ) \leq 1.1 \ell_1 + 2 \ell_2$:] The value of $\RNA(\cdot)$ can be upper bounded by the number of $1'$ and $0'$, which is $(2 \ell_2 + \ell_{\CLG, 1}) + (\ell_1 +　\ell_{\CLG, 0} ) \leq  1.1 \ell_1 + 2 \ell_2$.
\item[(6)~$\RNA( 0^{\ell_4} \CG_\beta(t) 0^{\ell_4} \CG_\gamma(t') ) \leq 1.1 \ell_1 + 4 \ell_2$:] Define the string $S = 0^{\ell_4} \circ \left( {1'}^{\ell_2}  {0'}^{\ell_1}  {1'}^{2\ell_2} {0'}^{\ell_1}{1'}^{\ell_2} \right) \circ   0^{\ell_4} \circ \left( {1}^{\ell_2}  {0}^{\ell_1}  {1}^{\ell_2} {0}^{\ell_1}  {1}^{2\ell_2} \right)$ as the result of removing the clique node gadgets and the clique list gadgets in $0^{\ell_4} \CG_\beta(t) 0^{\ell_4} \CG_\gamma(t')$. It is clear that $\RNA( 0^{\ell_4} \CG_\beta(t) 0^{\ell_4} \CG_\gamma(t') )$ $\leq 0.1 \ell_1 + \RNA(S)$, as the total length of the removed substrings can be upper bounded by $0.1 \ell_1$. Therefore, it suffices to show that $\RNA(S) \leq  \ell_1 + 4 \ell_2$.
    Let $A$ be any  RNA folding of $S$.

\begin{description}
\item[Case~1:] There is a pair $(0,0')\in A$ where the letter $0'$ comes from the first ${0'}^{\ell_1}$ in $S$. Clearly, the first substring $ {1'}^{\ell_2}$ cannot participate in any pairing. Therefore, $|A| \leq |{0'}^{\ell_1}  {1'}^{2\ell_2} {0'}^{\ell_1}{1'}^{\ell_2}| = 2 \ell_1 + 3 \ell_2 < \ell_1 + 4 \ell_2$.

\item[Case~2:] There is a pair $(0',0)\in A$ where the letter $0'$ comes from the first ${0'}^{\ell_1}$ in $S$.
Consider any pair $(1',1)$ such that the letter $1$ is in the substring ${1}^{2\ell_2}$ in $S$.
In order to have this pair not crossing any pair  $(0',0)\in A$, the letter $1'$ must be in the first substring ${1'}^{\ell_2}$ in $S$.
Therefore, at most half of the letters in the substring ${1}^{2\ell_2}$ can participate in the RNA folding $A$, and so $|A|$ is at most the total number of $0'$ and $1$ in $S$ minus $\ell_2$, i.e., $|A| \leq 2 \ell_1 + 3 \ell_2 < \ell_1 + 4 \ell_2$.

\item[Case~3:] The first ${0'}^{\ell_1}$ in $S$ does not participate in the RNA folding. In this case, we have $|A| \leq |{1'}^{\ell_2}  {1'}^{2\ell_2} {0'}^{\ell_1}{1'}^{\ell_2}| = \ell_1 + 4 \ell_2$.\qedhere
\end{description}
\end{description}
\end{proof}

Note that Lemma~\ref{lem-4}(1, 2, 3) implies that the RNA folding for blocked clique gadgets described in Fig.~\ref{fig-2} is optimal, and the optimal number of pairings is irrelevant to the underlying $k$-clique.

\subsection{Optimal RNA foldings of $S_G$ \label{ss-3}}

In this section, we show that there exists an optimal RNA folding of $S_G$ satisfying some good properties.
Let $A$ be an RNA folding of a string $S$, and let $S_1$ and $S_2$ be two disjoint substrings of $S$. Recall that a substring is a subsequence of consecutive elements. We write $S_1 \overset{A}\longleftrightarrow S_2$ if there exists a pair $\{x_1, x_2\} \in A$ with $x_1 \in S_1, x_2 \in S_2$.
Given an RNA folding $A$ of the string $S_G$, the two properties (P1) and (P2) are defined as follows.
\begin{description}
\item[(P1)] All $0'$ in all ${0'}^{\ell_3}$ are paired up with some  $0$ in some $0^{\ell_4}$ in $A$.
\item[(P2)] There exist $t_{\alpha}, t_{\beta}, t_{\gamma} \in \mathcal{C}_k$ such that the following holds. If $\CG_{u_1}(t_1) \overset{A}\longleftrightarrow \CG_{u_2}(t_2)$, then $\{(u_1, t_1), (u_2,t_2)\} \subseteq \{(\alpha,t_{\alpha}),(\beta, t_\beta), (\gamma, t_\gamma)\}$.
\end{description}
Intuitively, (P2) says that all clique gadgets are blocked by the  pairings between ${0'}^{\ell_3}$ and $0^{\ell_4}$, except the three selected clique gadgets  $\CG_\alpha(t_\alpha)$, $\CG_\beta(t_\beta)$, and $\CG_\gamma(t_\gamma)$, for some choices of three $k$-cliques $t_\alpha$, $t_\beta$, and $t_\gamma$.

\begin{lemma}\label{lem-6}
Let $A$ be any RNA folding of $S_G$.
Let $S_1$ be a substring ${0'}^{\ell_3}$ of $S_G$, and let $S_2$ be a substring $0^{\ell_4}$ of $S_G$.
If there is a pair in $A$ linking a letter $0'$ in $S_1$ to a letter $0$ in $S_2$, then there is another RNA folding $A'$ of $S_G$ with $|A'| \geq |A|$ where all letters in $S_1$  are paired up with letters in $S_2$.
\end{lemma}
\begin{proof}
The lemma immediately follows from the fact that $\ell_4$ is greater than the total number of $0'$ in $S_G$, which makes it possible to rematch all the letters in $S_1$ to  letters in $S_2$.
\end{proof}

Lemma~\ref{lem-7} shows that there is an optimal RNA folding $A$ of $S_G$ satisfying (P1).

\begin{lemma}\label{lem-7}
There is an optimal RNA folding of $S_G$ satisfying (P1).
\end{lemma}
\begin{proof}
Choose any RNA folding $A$ of $S_G$ with $|A| = \RNA(S_G)$. In view of  Lemma~\ref{lem-6}, we assume that for each substring ${0'}^{\ell_3}$ in $S_G$, either (i) all its letters are matched to letters in the same $0^{\ell_4}$, or (ii) none of its letters is matched to any letters in any $0^{\ell_4}$. Let $z$ be the number of ${0'}^{\ell_3}$ such that none of its letters is matched to any letters in any $0^{\ell_4}$.

Let $t \in \mathcal{C}_k$, and let $x \in \{\alpha, \beta, \gamma\}$. We say that $\CG_x(t)$ is {\em trapped} in $A$ if each letter in $\CG_x(t)$ is either (i) unmatched, (ii) matched to letters within $\CG_x(t)$, or (iii) matched to a letter in some $0^{\ell_4}$.
Note that a sufficient condition for a clique gadget $\CG_x(t)$ to be trapped is that all letters in its two neighboring ${0'}^{\ell_3}$  are  matched to letters in the same substring $0^{\ell_4}$.

Suppose that the clique gadget $\CG_x(t)$ is not trapped in $A$, then $\CG_x(t)$ falls into one of the following two cases.
\begin{description}
\item [Case 1:] The letters in the two neighboring substrings ${0'}^{\ell_3}$ of $\CG_x(t)$ are matched to letters in two distinct substrings $0^{\ell_4}$.
\item [Case 2:] A neighboring ${0'}^{\ell_3}$ of $\CG_x(t)$ is not matched to any $0^{\ell_4}$.
\end{description}
Observe that at most  $3$ clique gadgets belong to the first case, and at most $2z$  clique gadgets belong to the second case. Thus, the number of clique gadgets that are not trapped in $A$ is at most $3 + 2z$.
We derive an upper bound of $|A|$ as follows.
\begin{align*}
|A| & \leq (3(|\mathcal{C}_k|+ 1 ) - z) \ell_3  & (\text{matched }{0'}^{\ell_3}) \\
&\hspace{0.4cm}+
|\mathcal{C}_k| \bigg(
\max_{t \in \mathcal{C}_k } \RNA( 0^{\ell_4} \CG_\alpha(t) )
+ \max_{t\in \mathcal{C}_k}  \RNA( 0^{\ell_4} \CG_\beta(t) ) & (\text{trapped clique gadgets})  \\
& \hspace{1.4cm} +  \max_{t\in \mathcal{C}_k}  \RNA( 0^{\ell_4} \CG_\gamma(t) ) \bigg)\\
&\hspace{0.4cm} + (3 + 2z) \max_{t\in \mathcal{C}_k , x \in \{\alpha, \beta, \gamma\}} |\CG_x(t)|. & (\text{remaining clique gadgets})
\end{align*}
In view of the calculation in Lemma~\ref{lem-4}, $|A|$ is at most
$$m_1 - z \ell_3 +
\left( 2 \ell_2 + \min (\ell_{\CLG, 1}, \ell_{\CNG, 1}) + 2 \ell_1 + \ell_{\CLG, 0} + \ell_{\CNG, 0} \right) +
(3 +2z) \max_{t, x}|\CG_x(t)|.$$
Due to the two facts  (i) $ 2 \ell_2 + \min (\ell_{\CLG, 1}, \ell_{\CNG, 1})  + 2 \ell_1  + \ell_{\CLG, 0} + \ell_{\CNG, 0}  < 0.1\ell_3$, and (ii) the length of a clique gadget $< 0.1\ell_3$, we have:
$$|A| <  m_1 - 0.8 z \ell_3 + 0.4 \ell_3.$$
Thus, if $z>0$, then $|A| < m_1 < \RNA(S_G)$, contradicting the assumption that $A$ is optimal. Hence we must have $z = 0$, i.e., $A$ satisfies (P1).
\end{proof}

Next, we deal with property (P2). We need some terminologies for ease of notation.
For each $x \in \{\alpha, \beta, \gamma\}$, we call $\CG_x(t)$ a {\em type-$x$} clique gadget.
We say that the two clique gadgets $C_1$ and $C_2$ are {\em linked} in $A$ if $C_1 \overset{A}\longleftrightarrow C_2$.
We write $\mathcal{M}_{\alpha, \beta, \gamma}$ to denote the set of all RNA foldings $A$ of $S_G$ satisfying (P1) and (P2).
We write $\mathcal{M}_\alpha$ to denote the set of all RNA foldings $A$ of $S_G$ satisfying (P1) and the following condition ($\text{P2}_\alpha'$).
\begin{description}
\item[($\text{P2}_\alpha'$)] There exist $t_{\alpha,1}, t_{\alpha,2}, t_{\beta}, t_{\gamma} \in \mathcal{C}_k$ satisfying $t_{\alpha,1} \neq t_{\alpha,2}$ such that the following holds. If $\CG_{u_1}(t_1) \overset{A}\longleftrightarrow \CG_{u_2}(t_2)$, then $\{(u_1, t_1), (u_2,t_2)\} \in \{ \{(\alpha,t_{\alpha,1}),(\beta, t_\beta)\},\{(\alpha,t_{\alpha,2}),(\gamma, t_\gamma)\} \}$.
\end{description}
The two properties ($\text{P2}_\beta'$) and ($\text{P2}_\gamma'$), and the two sets $\mathcal{M}_\beta$ and $\mathcal{M}_\gamma$ are defined analogously.

\begin{lemma}\label{lem-8}
Let $A$ be an optimal RNA folding of $S_G$ satisfying (P1). For each $x \in \{\alpha, \beta, \gamma\}$, there do not exist two distinct type-$x$ clique gadgets $C_1$ and $C_2$ with $C_1 \overset{A}\longleftrightarrow C_2$.
\end{lemma}
\begin{proof}
There is a substring $S^\diamond = {0'}^{\ell_3}$ located between $C_1$ and $C_2$. The existence of a pair in $A$ linking a letter in $C_1$ and a letter in $C_2$ makes it impossible for any letter in $S^\diamond$ be matched to letters in any $0^{\ell_4}$, which is a contradiction to the assumption that $A$ has property (P1).
\end{proof}

\begin{lemma}\label{lem-9}
Let $A$ be an optimal RNA folding of $S_G$ satisfying (P1). For each $\{x,y\} \in \{ \{\alpha, \beta\}$, $\{\alpha, \gamma\}, \{\beta, \gamma\}\}$, there do not exist two distinct type-$x$ clique gadgets $C_1$ and $C_2$ and two not necessarily distinct type-$y$ clique gadgets $C_3$ and $C_4$ such that $C_1 \overset{A}\longleftrightarrow C_3$ and $C_2 \overset{A}\longleftrightarrow C_4$.
\end{lemma}
\begin{proof}
There is a substring $S^\diamond = {0'}^{\ell_3}$ located between $C_1$ and $C_2$. Since $C_1 \overset{A}\longleftrightarrow C_3$ and $C_2 \overset{A}\longleftrightarrow C_4$, letters in $S^\diamond$ can only be matched to (i) letters in  $C_1$, $C_2$, $C_3$, and $C_4$, (i) letters located between  $C_1$ and $C_2$, and (iii) letters located between $C_3$ and $C_4$. This contradicts the assumption that $A$ has property (P1).
\end{proof}

\begin{lemma}\label{lem-10}
Let $A$ be an optimal RNA folding of $S_G$ satisfying (P1).
Let $x \in \{\alpha, \beta, \gamma\}$.
Suppose that there are two distinct type-$x$ clique gadgets $C_1$ and $C_2$  such that $C_1 \overset{A}\longleftrightarrow C_3$ and $C_2 \overset{A}\longleftrightarrow C_4$, where  $C_3$ and $C_4$ are two not necessarily distinct clique gadgets.
Then $A \in \mathcal{M}_x$.
\end{lemma}
\begin{proof}
Suppose that $C_3$ is a type-$y$ clique gadget, and $C_4$ is a type-$z$ clique gadget. By Lemma~\ref{lem-8} and Lemma~\ref{lem-9}, the three symbols $x$, $y$, and $z$ must be distinct, and so $C_3 \neq C_4$.

To prove that  $A \in \mathcal{M}_x$, it suffices to show that $A$ satisfies ($\text{P2}_x'$).
Suppose that ($\text{P2}_x'$) is not met, then there are two clique gadgets $C_5$ and $C_6$ that are linked in $A$, and $\{C_5, C_6\} \notin \{\{C_1, C_3\},\{C_2, C_4\}\}$. We show that this leads to a contradiction.

Observe that none of  $C_5$ and $C_6$ can be a type-$x$ clique gadget. Suppose that $C_5$ is of type-$x$. Then $C_6$ is either type-$y$ or type-$z$ by Lemma~\ref{lem-8}. In any case, Lemma~\ref{lem-9} is violated.
Therefore, without loss of generality, we assume $C_5$ is of type-$y$. Then, by Lemma~\ref{lem-8}, $C_6$ must be of type-$z$.

Since $C_1$ and $C_2$ are distinct, there must be a substring $S^\diamond = {0'}^{\ell_3}$ located between $C_1$ and $C_2$.  Since $C_1$ is linked to a type-$y$ clique gadget, and since $C_2$ is linked to a type-$z$ clique gadget, letters in $S^\diamond$ can only be paired up with letters in the substring $S' = 0^{\ell_4}$ bordering both ${0'}^{\ell_3}  \underset{{t \in \mathcal{C}_k}}{\bigcirc}  \left(   \CG_y(t) {0'}^{\ell_3} \right)$ and ${0'}^{\ell_3}  \underset{{t \in \mathcal{C}_k}}{\bigcirc}  \left(  \CG_z(t) {0'}^{\ell_3} \right)$, viewing $S_G$ as a circular string. However, the existence of a pair linking a letter in $C_5$ (which is of type-$y$) and a letter in $C_6$ (which is of type-$z$) implies that no letter in $S'$ can be matched with a letter in ${0'}^{\ell_3}$ without a crossing.  This contradicts the assumption that $A$ has property (P1).
\end{proof}

The following lemma  follows from Lemma~\ref{lem-7} and Lemma~\ref{lem-10}.

\begin{lemma}\label{lem-11}
There is an optimal RNA folding of $S_G$ that belongs to  $\mathcal{M}_\alpha \cup \mathcal{M}_\beta \cup \mathcal{M}_\gamma \cup \mathcal{M}_{\alpha, \beta, \gamma}$.
\end{lemma}
\begin{proof}
By Lemma~\ref{lem-7}, we restrict our consideration to optimal RNA foldings having (P1), and let $A$ be any such optimal RNA folding.
There are two cases.
\begin{description}
\item[Case 1:]  For each $x \in \{\alpha, \beta, \gamma\}$, there is at most one type-$x$ clique gadget that is linked to other clique gadgets. Then $A$ satisfies (P2), and so $A \in \mathcal{M}_{\alpha, \beta, \gamma}$.
\item[Case 2:]  For some $x \in \{\alpha, \beta, \gamma\}$, there are two distinct type-$x$ clique gadgets that are linked to other clique gadgets. By Lemma~\ref{lem-10}, $A \in \mathcal{M}_x$.\qedhere
\end{description}
\end{proof}

We are now in a position to prove the main lemma in this section, which shows that there is an optimal RNA folding satisfying (P1) and (P2).

\begin{lemma}\label{lem-12}
There is an optimal RNA folding of $S_G$ that belongs to $\mathcal{M}_{\alpha, \beta, \gamma}$.
\end{lemma}
\begin{proof}
In view of Lemma~\ref{lem-11}, it suffices to show that for any $A \in \mathcal{M}_\alpha \cup \mathcal{M}_\beta \cup \mathcal{M}_\gamma$, we have $|A| < \RNA(S_G)$.
Let $A \in \mathcal{M}_x$.
Let $t_{x,1}, t_{x,2}, t_{y}, t_{z} \in \mathcal{C}_k$ and  $\{y,z\} = \{\alpha, \beta, \gamma\} \setminus \{x\}$ be the parameters specified in the definition of $\mathcal{M}_x$.  That is, if $\CG_{u_1}(t_1) \overset{A}\longleftrightarrow \CG_{u_2}(t_2)$, then $\{(u_1, t_1), (u_2,t_2)\} \in \{ \{(\alpha,t_{\alpha,1}),(\beta, t_\beta)\},\{(\alpha,t_{\alpha,2}),(\gamma, t_\gamma)\} \}$.
 Each pair in $A$ falls into one of the following cases.
\begin{description}
\item[Case 1:] The pair links a letter $0'$ in some substring ${0'}^{\ell_3}$ to a letter $0$ in some substring $0^{\ell_4}$. Note that there are exactly $3(|\mathcal{C}_k|+1) \ell_3$ number of such pairs in $A$.
\item[Case 2:] The pair contains a letter in some $\CG_{u}(t)$, where $(u,t) \notin \{(x,t_{x,1}),(x,t_{x,2}), (y, t_y), (z,t_z)\}$. Observe that any letter in such $\CG_{u}(t)$ can only be matched to letters within the same clique gadget $\CG_{u}(t)$ or some substring $0^{\ell_4}$, and so the number of pairs in $A$ that fall into this case is upper bounded by
    $(|\mathcal{C}_k| - 2) \max\limits_{t \in \mathcal{C}_k}\RNA\left(0^{\ell_4} \CG_{x}(t)\right)
    + (|\mathcal{C}_k| - 1) \max\limits_{t \in \mathcal{C}_k} \RNA\left(0^{\ell_4} \CG_{y}(t)\right)
    + (|\mathcal{C}_k| - 1) \max\limits_{t \in \mathcal{C}_k} \RNA\left(0^{\ell_4} \CG_{z}(t)\right)$.
\item[Case 3:]  The pair involves a letter in some $\CG_{u}(t)$, where $(u,t) \in \{(x,t_{x,1}),(x,t_{x,2}), (y, t_y), (z,t_z)\}$. The number of such pairs is upper bounded by
    $\max\limits_{t, t' \in \mathcal{C}_k} \RNA\left(0^{\ell_4} \CG_{x}(t) 0^{\ell_4} \CG_{y}(t')\right)
    + \max\limits_{t, t' \in \mathcal{C}_k} \RNA$ $\left(0^{\ell_4} \CG_{x}(t) 0^{\ell_4} \CG_{z}(t')\right)$.
\end{description}
Using the above calculation and the formulas in Lemma~\ref{lem-4}, we can derive $|A|  < m_1 + m_2$, as follows. Note that $m_2 \geq 6 \ell_2 + 3 \ell_1$.
\[
|A| \leq
\begin{cases}
m_1 + 2 \ell_2 + 4.2 \ell_1 - \min (\ell_{\CLG, 1}, \ell_{\CNG, 1})
&\text{ if $x = \alpha$}\\
m_1 + 6 \ell_2 + 2.2 \ell_1 - \ell_{\CLG, 0} - \ell_{\CNG, 0}
&\text{ if $x = \beta$}\\
m_1 + 6 \ell_2 + 2.2 \ell_1
&\text{ if $x = \gamma$}
\end{cases}
\]
Note that $m_1 + m_2 \leq \RNA(S_G)$ by Lemma~\ref{lem-5}.
\end{proof}

\subsection{Calculating $\RNA(S_G)$ \label{ss-4}}

In this section, we prove that $\RNA(S_G) = m_1 + m_2$, and finish the proof of Theorem~\ref{thm-2}.
In view of Lemma~\ref{lem-12}, in the calculation of $\RNA(S_G)$, we can restrict our consideration to RNA foldings in $\mathcal{M}_{\alpha, \beta, \gamma}$. Based on the structural property of RNA foldings in $\mathcal{M}_{\alpha, \beta, \gamma}$, we first reduce the calculation of $\RNA(S_G)$ to the calculation of optimal RNA foldings of simpler strings.

\begin{lemma}\label{lem-13}
$\RNA(S_G) \leq m_1 + \max\limits_{t_{\alpha}, t_{\beta}, t_{\gamma} \in \mathcal{C}_k} \RNA(0^{\ell_4} \CG_\alpha(t_\alpha) 0^{\ell_4} \CG_\beta(t_\beta) 0^{\ell_4} \CG_\gamma(t_\gamma))$.
\end{lemma}
\begin{proof}
By Lemma~\ref{lem-12}, there is an optimal RNA folding of $S_G$ in $\mathcal{M}_{\alpha, \beta, \gamma}$.
For any $A \in \mathcal{M}_{\alpha, \beta, \gamma}$, let $t_{\alpha}, t_{\beta}, t_{\gamma} \in \mathcal{C}_k$ be the three cliques in the definition of (P2). Each pair in $A$ falls into one of the following cases.
\begin{description}
\item[Case 1:] The pair links a letter $0'$ in some substring ${0'}^{\ell_3}$ to a letter $0$ in some substring $0^{\ell_4}$. Note that there are exactly $3(|\mathcal{C}_k|+1) \ell_3$ number of such pairs in $A$.
\item[Case 2:]
The pair contains a letter in some $\CG_{u}(t)$,  where $(u,t) \notin \{(\alpha,t_{\alpha}),(\beta,t_{\beta}), (\gamma, t_\gamma)\}$. Observe that any letter in such $\CG_{u}(t)$ can only be matched to letters within the same clique gadget $\CG_{u}(t)$ or some substring $0^{\ell_4}$, and so the number of pairs in $A$ that fall into this case is upper bounded by
    $(|\mathcal{C}_k| - 1) \max\limits_{t \in \mathcal{C}_k}\RNA(0^{\ell_4} \CG_{\alpha}(t))
    + (|\mathcal{C}_k| - 1) \max\limits_{t \in \mathcal{C}_k} \RNA(0^{\ell_4} \CG_{\beta}(t))
    + (|\mathcal{C}_k| - 1) \max\limits_{t \in \mathcal{C}_k} \RNA(0^{\ell_4} \CG_{\gamma}(t))$.
\item[Case 3:] The pair involves a letter in some $\CG_{u}(t)$, where $(u,t) \in \{(\alpha,t_{\alpha}),(\beta,t_{\beta}), (\gamma, t_\gamma)\}$. The number of such  pairs is upper bounded by
    $\RNA(0^{\ell_4} \CG_\alpha(t_\alpha) 0^{\ell_4} \CG_\beta(t_\beta) 0^{\ell_4} \CG_\gamma(t_\gamma))$.
\end{description}
Applying the formulas in Lemma~\ref{lem-4}, we have $|A| = m_1 + \RNA(0^{\ell_4} \CG_\alpha(t_\alpha) 0^{\ell_4} \CG_\beta(t_\beta) 0^{\ell_4} \CG_\gamma(t_\gamma))$. Hence we conclude the proof.
\end{proof}

For any choices of three $k$-cliques  $t_{\alpha}, t_{\beta}, t_{\gamma} \in \mathcal{C}_k$, we define:
$$S_{t_{\alpha}, t_{\beta}, t_{\gamma}} = {1}^{\ell_2} \circ S_{t_{\gamma}, t_{\alpha}} \circ  {1}^{\ell_2} \circ S_{t_{\alpha}, t_{\beta}, }  \circ {1'}^{2 \ell_2}  \circ S_{t_{\beta}, t_{\gamma}},$$
where
\begin{align*}
S_{t_{\gamma}, t_{\alpha}} &= {0}^{\ell_1} \CNG(t_\gamma) p({\CLG(t_\alpha)}^R) {0'}^{\ell_1},\\
S_{t_{\alpha}, t_{\beta}} &= {0}^{\ell_1} \CNG(t_\alpha) p({\CLG(t_\beta)}^R) {0'}^{\ell_1},\\
S_{t_{\beta}, t_{\gamma}} &= {0'}^{\ell_1} p(\CNG(t_\beta)) {\CLG(t_\gamma)}^R {0}^{\ell_1}.\\
\end{align*}
Note that $S_{t_{\alpha}, t_{\beta}, t_{\gamma}}$ is simply a cyclic shift of the concatenation of $\CG_\alpha(t_\alpha)$, $\CG_\beta(t_\beta)$, and $\CG_\gamma(t_\gamma)$ after removing the sequences of $1$s and $1'$s at the beginning and the end of these clique gadgets.
Lemma~\ref{lem-15}, together with Lemma~\ref{lem-13}, reduces the calculation of $\RNA(S_G)$ to the calculation of $\RNA(S_{t_{\alpha}, t_{\beta}, t_{\gamma}})$.
Lemma~\ref{lem-14} and Lemma~\ref{lem-15-aux} are auxiliary lemmas.

\begin{lemma}\label{lem-14}
Let $S = S_1 \circ S_2 \circ S_3 \in \{0,1,0',1'\}^{\ast}$ be a string, where the substring $S_2$ is either $1 1'$ or $1' 1$. Then $\RNA(S) = \RNA(S_1 \circ S_3) + 1$.
\end{lemma}
\begin{proof}
It suffices to show that there exists an optimal RNA folding of $S$  where the two letters in $S_2$ are paired up.
Let $A$ be any optimal RNA folding of $S$. We show that it is possible to modify $A$ in such a way that the two letters in $S_2$ are paired up, and the total number of matched pairs is unchanged.
\begin{description}
\item[Case 1:] If the two letters in $S_2$ are already paired up, then no modification is needed.
\item[Case 2:] If exactly one of the two letters in $S_2$ is matched in $A$, we first unmatch it, and then we pair up the two letters in $S_2$.
\item[Case 3:] Suppose that both  two letters in $S_2$ are matched to letters not in $S_2$. Suppose that the letter $1 \in S_2$ is paired up with $x$ in $A$, and the letter $1' \in S_2$ is paired up with $y$ in $A$. We removing these two pairs from $A$,  and then we add the two pairs $\{x,y\}$ and $\{1,1'\}$ to $A$.\qedhere
\end{description}
\end{proof}

\begin{lemma}\label{lem-15-aux}
In any optimal RNA folding of $0^{\ell_4} \CG_\alpha(t_\alpha) 0^{\ell_4} \CG_\beta(t_\beta) 0^{\ell_4} \CG_\gamma(t_\gamma)$, no letter within the three substrings $0^{\ell_4}$ is matched.
\end{lemma}
\begin{proof}
For ease of notation, we write $S^\star$ to denote $0^{\ell_4} \CG_\alpha(t_\alpha) 0^{\ell_4} \CG_\beta(t_\beta) 0^{\ell_4} \CG_\gamma(t_\gamma)$.
We first state a few simple observations.
\begin{description}
\item[(O1)] By simply matching only the letters in $\CG_\alpha(t_\alpha), \CG_\beta(t_\beta)$, and $\CG_\gamma(t_\gamma)$, as described in Fig.~\ref{fig-3}, we  infer that $\RNA(S^\star) \geq 6 \ell_2 + 3 \ell_1$.
\item[(O2)]
The total number of $0'$ and $1'$ in $S^\star$ is at most $6 \ell_2 + 3.1 \ell_1$, and so
$\RNA(S^\star) \leq 6 \ell_2 + 3.1 \ell_1$.
\item[(O3)]
The absolute difference between the number of $1$ and the number of $1'$ in $S^\star$ is at most $0.1 \ell_1$.
\end{description}

Note that $\CG_\gamma(t_\gamma)$ does not contain $0'$.
To prove the lemma, we show that in any optimal RNA folding $A$ of $S^\star$, no letter in $0^{\ell_4}$ is matched to letters in $\CG_\beta(t_\beta)$ and $\CG_\alpha(t_\alpha)$.
We write $S_1$, $S_2$, and $S_3$ to denote the first, second, and the third substring $0^{\ell_4}$.

We show that there is no pair in $A$ linking a letter $0' \in \CG_\beta(t_\beta)$ to a letter $0 \in S_2 \cup S_3$.  Recall that $\CG_\beta(t_\beta) = {1'}^{\ell_2} p({\CLG(t_\beta)}^R) {0'}^{\ell_1}  {1'}^{2\ell_2} {0'}^{\ell_1} p(\CNG(t_\beta)) {1'}^{\ell_2}$. If there is such a pair, then at least $\ell_2$ amount of $1'$ cannot participate in the RNA folding. Therefore, by (O2), $|A| \leq (6 \ell_2 + 3.1 \ell_1) - \ell_2 \leq
5 \ell_2 + 3.1 \ell_1$.
However, by (O1), $\RNA(S^\star) \geq 6 \ell_2 + 3 \ell_1 > |A|$, contradicting with the assumption that $A$ is optimal.

Next, we show that there is no pair in $A$ linking a letter $0' \in \CG_\beta(t_\beta)$ to a letter $0 \in S_1$.
Recall that $\CG_\alpha(t_\alpha) = {1'}^{2 \ell_2} p({\CLG(t_\alpha)}^R) {0'}^{\ell_1}  {1}^{\ell_2} {0}^{\ell_1} \CNG(t_\alpha) {1}^{\ell_2}$.
 Suppose that there is such a pair. Then the $3\ell_2$ amount of $1'$ in (i) the substring ${1'}^{2\ell_2}$ of $\CG_\alpha(t_\alpha)$ and (ii) the  first substring ${1'}^{\ell_2}$ of $\CG_\beta(t_\beta)$ can only be matched to letters in $\CG_\alpha(t_\alpha)$. However, the amount of $1$ in $\CG_\alpha(t_\alpha)$ is at most $2.1 \ell_1$, and so at least $0.9 \ell_2$ amount of $1'$ are not matched. Therefore, by (O1) and (O2), $|A| \leq (6 \ell_2 + 3.1 \ell_1) - 0.9 \ell_2 <  \RNA(S^\star)$, contradicting with the assumption that $A$ is optimal.

Lastly, we show that there is no pair in $A$ linking a letter $0' \in \CG_\alpha(t_\alpha)$  to a letter $0  \in S_1 \cup S_2 \cup S_3$. Suppose that there is such a pair. We show that at least $\ell_2$ amount of $1'$ cannot participate in the RNA folding. Then, by (O1) and (O2), $|A| \leq (6 \ell_2 + 3.1 \ell_1) - \ell_2 <  \RNA(S^\star)$, contradicting with the assumption that $A$ is optimal. We divide the analysis into cases.
\begin{description}
    \item[Case 1:] A letter $0' \in \CG_\alpha(t_\alpha)$ is matched to a letter $0 \in S_1$. Then all letters in the substring ${1'}^{2 \ell_2}$ of  $\CG_\alpha(t_\alpha)$  cannot participate in the RNA folding.
    \item[Case 2:] A letter $0' \in \CG_\alpha(t_\alpha)$ is matched to a letter $0 \in S_2$. Then all letters in the two substrings ${1}^{\ell_2}$ of $\CG_\alpha(t_\alpha)$  can only be matched to letters within $p(\text{CLG}(t_\alpha)^R)$. Therefore, at least
        $2 \ell_2 - |p(\text{CLG}(t_\alpha)^R)| >  2 \ell_2 - 0.1 \ell_1$ amount of $1$ are unmatched. By (O3), the absolute difference between the number of $1$ and the number of $1'$ in $S^\star$ is at most $0.1 \ell_1$, and so at least $2 \ell_2 - 0.2 \ell_1 > \ell_2$ amount of $1'$ cannot participate in the RNA folding.
     \item[Case 3:] A letter $0' \in \CG_\alpha(t_\alpha)$ is matched to a letter $0 \in S_3$. Then all  $1'$ in $\CG_\beta(t_\beta)$ can only be matched to letters in $\CG_\alpha(t_\alpha)$. Observe that the number of $1'$ in $\CG_\beta(t_\beta)$ is at least $\ell_2$ more than the number of  $1$ in $\CG_\alpha(t_\alpha)$, and so at least $\ell_2$ amount of $1'$  cannot participate in the RNA folding.\qedhere
\end{description}
\end{proof}

\begin{lemma}\label{lem-15}
$\RNA(0^{\ell_4} \CG_\alpha(t_\alpha) 0^{\ell_4} \CG_\beta(t_\beta) 0^{\ell_4} \CG_\gamma(t_\gamma)) = 4 \ell_2 + \RNA(S_{t_{\alpha}, t_{\beta}, t_{\gamma}})$.
\end{lemma}
\begin{proof} We bound $\RNA(0^{\ell_4} \CG_\alpha(t_\alpha) 0^{\ell_4} \CG_\beta(t_\beta) 0^{\ell_4} \CG_\gamma(t_\gamma))$ as follows.
\begin{align*}
&\RNA(0^{\ell_4} \CG_\alpha(t_\alpha) 0^{\ell_4} \CG_\beta(t_\beta) 0^{\ell_4} \CG_\gamma(t_\gamma))\\
&\hspace{0.7cm} = \RNA(\CG_\alpha(t_\alpha)\CG_\beta(t_\beta) \CG_\gamma(t_\gamma)) & (\text{Lemma~\ref{lem-15-aux}})\\
&\hspace{0.7cm} = \RNA(
{1'}^{2 \ell_2} p({\CLG(t_\alpha)}^R) {0'}^{\ell_1}  {1}^{\ell_2} {0}^{\ell_1} \CNG(t_\alpha) {1}^{\ell_2}
{1'}^{\ell_2} p({\CLG(t_\beta)}^R){0'}^{\ell_1}  {1'}^{2\ell_2} & (\text{by definition})\\
&\hspace{1.2cm}  {0'}^{\ell_1} p(\CNG(t_\beta)) {1'}^{\ell_2}
{1}^{\ell_2} {\CLG(t_\gamma)}^R {0}^{\ell_1}  {1}^{\ell_2} {0}^{\ell_1} \CNG(t_\gamma) {1}^{2\ell_2}
)\\
&\hspace{0.7cm} =  \RNA(
{1}^{\ell_2} {0}^{\ell_1} \CNG(t_\gamma) {1}^{2\ell_2}
{1'}^{2 \ell_2} p({\CLG(t_\alpha)}^R) {0'}^{\ell_1}  {1}^{\ell_2} {0}^{\ell_1} \CNG(t_\alpha){1}^{\ell_2}
{1'}^{\ell_2} & (\text{cyclic shift})\\
&\hspace{1.2cm}  p({\CLG(t_\beta)}^R) {0'}^{\ell_1}  {1'}^{2\ell_2}
{0'}^{\ell_1} p(\CNG(t_\beta)) {1'}^{\ell_2}
{1}^{\ell_2} {\CLG(t_\gamma)}^R {0}^{\ell_1}
)\\
&\hspace{0.7cm} = 4 \ell_2 +  \RNA(
{1}^{\ell_2} {0}^{\ell_1} \CNG(t_\gamma)  p({\CLG(t_\alpha)}^R) {0'}^{\ell_1}  {1}^{\ell_2} {0}^{\ell_1} \CNG(t_\alpha)p({\CLG(t_\beta)}^R) {0'}^{\ell_1} & (\text{Lemma~\ref{lem-14}})\\
&\hspace{1.2cm}  {1'}^{2\ell_2}
{0'}^{\ell_1} p(\CNG(t_\beta))  {\CLG(t_\gamma)}^R {0}^{\ell_1}
)\\
&\hspace{0.7cm} = 4 \ell_2 + \RNA(S_{t_{\alpha}, t_{\beta}, t_{\gamma}}). & (\text{by definition})
\end{align*}
For the third equality, we  move ${1}^{\ell_2} {0}^{\ell_1} \CNG(t_\gamma) {1}^{2\ell_2}$ from the end of the sequence to the beginning.
The fourth equality follows by applying Lemma~\ref{lem-14} iteratively to remove the substrings ${1}^{2\ell_2} {1'}^{2\ell_2}$, ${1}^{\ell_2} {1'}^{\ell_2}$, and ${1'}^{\ell_2} {1}^{\ell_2}$.
\end{proof}

Lemma~\ref{lem-16} shows that $\RNA(S_G) = m_1 + m_2$ by calculating the exact value of $\RNA(S_{t_{\alpha}, t_{\beta}, t_{\gamma}})$.

\begin{lemma}\label{lem-16}
$\RNA(S_G) = m_1 + m_2$.
\end{lemma}
\begin{proof}
By Lemma~\ref{lem-5}, we already have $\RNA(S_G) \geq m_1 + m_2$.
By Lemma~\ref{lem-13} and Lemma~\ref{lem-15}, we have
$\RNA(S_G) \leq m_1 + 4 \ell_2  +
\max\limits_{t_{\alpha}, t_{\beta}, t_{\gamma} \in \mathcal{C}_k} \RNA(S_{t_{\alpha}, t_{\beta}, t_{\gamma}})$.
Thus, to prove the lemma, it suffices to show that $\RNA(S_{t_{\alpha}, t_{\beta}, t_{\gamma}}) = 2 \ell_2 + 3 \ell_1 + \frac{3}{2} \ell_0
 - Q(t_{\alpha}, t_{\beta}, t_{\gamma})$. Recall that $Q(t_{\alpha}, t_{\beta}, t_{\gamma})$ is defined as
 $\frac{1}{2}\big(
 \dLCS (\CLG(t_\alpha), \CNG(t_\beta)) +
 \dLCS(\CLG(t_\alpha), \CNG(t_\gamma)) +
 \dLCS (\CLG(t_\beta), \CNG(t_\gamma))
 \big)$.

First of all, we calculate a simple lower bound on $\RNA(S_{t_{\alpha}, t_{\beta}, t_{\gamma}})$.
Pairing up letters not residing in clique node gadgets and clique list gadgets yields an RNA folding of $S_{t_{\alpha}, t_{\beta}, t_{\gamma}}$ with cardinality $2\ell_2 + 3 \ell_1$, and so $\RNA(S_{t_{\alpha}, t_{\beta}, t_{\gamma}}) \geq 2 \ell_2 + 3 \ell_1$.
We claim that for any optimal RNA folding $A$ of $S_{t_{\alpha}, t_{\beta}, t_{\gamma}}$, the following two statements are true.
\begin{description}
\item[(S1)] For each of the two substrings ${1}^{\ell_2}$, there is a letter $1$ paired up with a letter $1'$ in the substring ${1'}^{2\ell_2}$.
\item[(S2)] For each $S' \in \{S_{t_{\gamma}, t_{\alpha}}, S_{t_{\alpha}, t_{\beta}}, S_{t_{\beta}, t_{\gamma}}\}$, there is a pair linking a letter $0'$ in ${0'}^{\ell_1} \subseteq S'$ and a letter $0$ in ${0}^{\ell_1} \subseteq S'$.
\end{description}

To prove the  statement (S1), suppose that a substring ${1}^{\ell_2}$ does not have any letter matched to a letter in the substring ${1'}^{2\ell_2}$. We show that this leads to a contradiction.
Observe that the number of $1'$ in $S_{t_{\alpha}, t_{\beta}, t_{\gamma}}$ that does not belong to  ${1'}^{2\ell_2}$ is at most $0.1 \ell_1$. Thus, $|A|$ is at most the total number of $0'$ plus the total number of $1$ minus $(\ell_2 - 0.1 \ell_1)$.
By a simple calculation, $|A| \leq (3 \ell_1 + 0.1 \ell_1) + (2 \ell_2 + 0.1 \ell_1) - (\ell_2 - 0.1 \ell_1) = \ell_2 + 3.3 \ell_1 < 2 \ell_2 + 3 \ell_1$, contradicting with the known lower bound of $\RNA(S_{t_{\alpha}, t_{\beta}, t_{\gamma}})$.

To prove the  statement (S2) suppose that there is a substring $S' \in \{S_{t_{\gamma}, t_{\alpha}}, S_{t_{\alpha}, t_{\beta}}, S_{t_{\beta}, t_{\gamma}}\}$ that has no pair linking a letter $0'$ in ${0'}^{\ell_1} \subseteq S'$ and a letter $0$ in ${0}^{\ell_1} \subseteq S'$. Due to (S1), any pair in $A$ involving letters in ${0'}^{\ell_1} \subseteq S'$ or ${0}^{\ell_1} \subseteq S'$ are confined to be within $S'$. Therefore, the number of pairs in $A$ involving letters in $S'$ is at most $|S| - 2\ell_1 \leq 0.1 \ell_1$.
This is certainly not optimal, since simply matching all $0'$ in ${0'}^{\ell_1}$ to all $0$ in ${0}^{\ell_1}$ gives us $\ell_1$ amount of pairs.

We can infer from the above two statements that for each $S' \in \{S_{t_{\gamma}, t_{\alpha}}, S_{t_{\alpha}, t_{\beta}}, S_{t_{\beta}, t_{\gamma}}\}$, letters within $S'$ are only matched to letters within $S$ in any optimal RNA folding of $S_{t_{\alpha}, t_{\beta}, t_{\gamma}}$.
As a result,
\begin{align*}
\RNA(S_{t_{\alpha}, t_{\beta}, t_{\gamma}})
& =  \RNA({1}^{\ell_2} \circ {1}^{\ell_2} \circ {1'}^{2\ell_2}) + \RNA(S_{t_{\gamma}, t_{\alpha}}) + \RNA(S_{t_{\alpha}, t_{\beta}}) +\RNA(S_{t_{\beta}, t_{\gamma}} )\\
&= 2 \ell_2 + 3 \ell_1 + \RNA(\CNG(t_\gamma) p({\CLG(t_\alpha)}^R))+ \RNA(\CNG(t_\alpha) p({\CLG(t_\beta)}^R)) \\
& \hspace{0.4cm} + \RNA(p(\CNG(t_\beta)) {\CLG(t_\gamma)}^R)\\
&=2 \ell_2 + 3 \ell_1 + \frac{3}{2} \ell_0
 - \frac{1}{2} \big(
 \dLCS (\CLG(t_\alpha), \CNG(t_\beta)) +
 \dLCS(\CLG(t_\alpha), \CNG(t_\gamma))\\
 & \hspace{0.4cm} +
 \dLCS (\CLG(t_\beta), \CNG(t_\gamma))
 \big).\qedhere
\end{align*}
\end{proof}

We are ready to prove Theorem~\ref{thm-2}.



\begin{proof}[Proof of Theorem~\ref{thm-2}]
Throughout the proof, $k$ is treated as a constant.
Given a graph $G$, we construct the  string $S_G$. According to Lemma~\ref{lem-1} and Lemma~\ref{lem-length}, the length of $S_G$ is ${O}( n^{k+1} \log n)$, and $S_G$ can be constructed in time ${O}\left( n^{k+1} \log n\right)$.
We let $t_{\alpha}, t_{\beta}, t_{\gamma} \in \mathcal{C}_k$ be chosen to minimize
\begin{align*}
&Q(t_{\alpha}, t_{\beta}, t_{\gamma})\\
&=\frac{1}{2}\big(
 \dLCS (\CLG(t_\alpha), \CNG(t_\beta)) +
 \dLCS(\CLG(t_\alpha), \CNG(t_\gamma)) +
 \dLCS (\CLG(t_\beta), \CNG(t_\gamma))
 \big).
 \end{align*}
By Lemma~\ref{lem-3}, there exists a number $c_1$ meeting the following conditions.
\begin{itemize}
 \item The number $c_1$ depends only on $n$ and $k$, and $Q(t_{\alpha}, t_{\beta}, t_{\gamma})  \geq 3 c_1 / 2$.
 \item If $Q(t_{\alpha}, t_{\beta}, t_{\gamma})  = 3 c_1 / 2$, then each of $t_{\alpha} \cup t_{\beta}$, $t_{\alpha} \cup t_{\gamma}$, and $t_{\beta} \cup t_{\gamma}$ is a $2k$-clique; in other words, $t_{\alpha} \cup t_{\beta} \cup t_{\gamma}$ is a $3k$-clique.
 \item If $Q(t_{\alpha}, t_{\beta}, t_{\gamma}) > 3c_1 / 2$, then the graph has no $3k$-clique.
\end{itemize}
According to Lemma~\ref{lem-16}, ${\RNA}(S_G) = m_1 + m_2$. By its definition, $m_1$ only depends on $n$ and $k$; and $m_2 = 6 \ell_2 + 3 \ell_1 + \frac{3}{2} \ell_0 - \min_{t_\alpha, t_\beta, t_\gamma \in \mathcal{C}_k} Q(t_\alpha, t_\beta, t_\gamma)$.  Hence we are able to infer whether  $G$ has a $3k$-clique from the value of ${\RNA}(S_G)$, which can be calculated in time $T\left({O}\left(n^{k+1} \log n\right)\right)$.
\end{proof}

\section{Hardness of Dyck Edit Distance Problem}

In this section, we consider the Dyck edit distance problem. The goal of this section is to present a simple reduction from RNA folding problem (with alphabet size 4) to Dyck edit distance problem (with alphabet size 10).

\paragraph{Dyck Edit Distance.}
Recall that the Dyck edit distance problem asks for the minimum number of edits to transform a given
string to a well-balanced parentheses of s different types.
The formal definition of the problem is as follows. Given $S \in (\Sigma \cup \Sigma')^n$, the goal of the Dyck edit distance problem is to find a minimum number of edit operations (insertion, deletion, and substitution) that transform $S$ into a string in the Dyck context free language.

Given $\Sigma$ and its corresponding $\Sigma'$, the Dyck context free language is defined by the grammar with following production rules: $\mathbf{S} \rightarrow \mathbf{SS}$, $\forall x \in \Sigma, \mathbf{S} \rightarrow x\mathbf{S}x'$, and $\mathbf{S} \rightarrow \epsilon$ (empty string). Note that for each $x \in \Sigma$, the two symbols $x$ and $x'$ represent one type of parenthesis.

\paragraph{An Alternate Formulation.} An alternative definition of the Dyck edit distance problem is  as follows.
Given a sequence $S \in (\Sigma \cup \Sigma')^n$, find a minimum cost set $A \subseteq \{(i,j) \ | \ 1 \leq i < j \leq n \}$ satisfying the following conditions:
\begin{itemize}
\item $A = A_M \cup A_S$ has no crossing pair.
\item $A_M$ contains only pairs of the form $(x,x')$, $x \in \Sigma$ (i.e. for all $(i,j)\in A_M$, we have $S[i]=x$, $S[j]=x'$, for some $x \in \Sigma$). $A_M$ corresponds to the set of matched pairs.
\item $A_S$ does not contain any pair of the form $(y',x)$, $x,y \in \Sigma$ (i.e. for all $(i,j)\in A_S$ we have either $S[i] \in \Sigma$ or $S[j] \in \Sigma'$). $A_S$ corresponds to the set of pairs that can be fixed by one substitution operation per each pair.
\item Let $D$ be the set of letters in $S$ that do not belong to any pair in $A$. Each letter in $D$ requires one deletion/insertion operation to fix.
\end{itemize}
The cost of $A$ is then defined as $|A_S|+|D|$, and the Dyck edit distance of the string $S$ is the cost of a minimum cost set meeting the above conditions.

\paragraph{Discussion.}
Dyck edit distance problem can be thought of as an asymmetric version of the RNA folding problem that also handles substitution, in addition to deletion and insertion. Though these two problems look similar, they can behave quite differently. For example, in Section~\ref{sec.intro} we describe a simple reduction from LCS to RNA folding; since LCS is basically the edit distance problem without substitution, one might feel that the same reduction also reduces the edit distance problem to the Dyck edit distance problem. However, this is not true. The following example shows that the edit distance between two strings $X$ and $Y$ cannot be inferred from the Dyck edit distance of $X \circ p(Y^R)$.
Both the two strings $X_1 = ababa$ and $X_2 = abbaa$ require  4 edit operations to transform into the string $Y = caaac$; but the Dyck edit distance of $X_1 \circ p(Y^R) = ababac'a'a'a'c'$ is 4 (by deleting all $b$ and $c'$), while the Dyck edit distance of $X_2 \circ p(Y^R) = abbaac'a'a'a'c'$ is 3 (by deleting all $c'$ and substituting the second $b$ with $b'$).

Intuitively, the substitution operation makes Dyck edit distance more complicated than RNA folding. Indeed,  the conditional lower bound  for Dyck edit distance shown in~\cite{ABV15} requires a big alphabet size of 48 and a longer proof.
In the remainder of this section, we prove Theorem~\ref{thm-3} by demonstrating a simple reduction from RNA folding problem (with alphabet size 4) to Dyck edit distance problem (with alphabet size 10). This improves upon the hardness result in \cite{ABV15}, and justifies the intuition that Dyck edit distance is at least as hard as RNA folding.




\begin{proof}[Proof of Theorem~\ref{thm-3}]
For notational simplicity, we let the alphabet for the RNA folding problem be $\Sigma \cup \Sigma'= \{0,0',1,1'\}$ instead of $\{A,C,G,U\}$. Let $S$ be any string in ${(\Sigma \cup \Sigma')^n}$. We define the string $S_{\text{Dyck}}$ as the result of applying the following operations on $S$:
\begin{itemize}
\item Replace each letter $0$ with the sequence $S_{0} = aeb'aeb'$.
\item Replace each letter $0'$ with the sequence $S_{0'} = bba'a'$.
\item Replace each letter $1$ with the sequence $S_{1} = ced'ced'$.
\item Replace each letter $1'$ with the sequence $S_{1'} = ddc'c'$.
\end{itemize}
It is clear that $S_{\text{Dyck}}$ is a sequence of length at most $6n$ on  the alphabet $\{a,b,c,d,e\} \cup \{a',b',c',d',e'\}$, though the letter $e'$ is not used. We claim that the Dyck edit distance of $S_{\text{Dyck}}$ is $\frac{|S_{\text{Dyck}}|}{2} - 2\RNA(S)$.

\paragraph{Upper Bound.} We show that the Dyck edit distance of $S_{\text{Dyck}}$ is at most $\frac{|S_{\text{Dyck}}|}{2} - 2\RNA(S)$. Given an optimal RNA folding of $S$, we construct a crossing-free matching $A$ with cost  $\frac{|S_{\text{Dyck}}|}{2} - 2\RNA(S)$ as follows.

\medskip

\noindent {For matched pairs in the RNA folding of $S$:}

\begin{itemize}
\item For each matched pair $(0,0')$ in the RNA folding of $S$,  we add two pairs $(a,a'),(a,a')$  to $A_M$, and add three pairs $(e,b'),(e,b'),(b,b)$ to $A_S$ in its corresponding pair of substrings $(S_0 = \mathbf{a}(eb')\mathbf{a}(eb'), S_{0'} = (bb)\mathbf{a'a'})$ in $S_{\text{Dyck}}$.
\item For each matched pair $(0',0)$ in the RNA folding of $S$, we add two pairs $(b,b'),(b,b')$ to $A_M$, and add three pairs $(a',a'),(a,e),(a,e)$ to $A_S$ in its corresponding pair of substrings $(S_{0'} = \mathbf{bb}(a'a'), S_{0} = (ae)\mathbf{b'}(ae)\mathbf{b'})$ in $S_{\text{Dyck}}$.
\item Similarly, for each matched pair $(1,1'),(1',1)$ in the RNA folding of $S$, we add two pairs to $A_M$ and three pairs to $A_S$.
\end{itemize}

\noindent {For unmatched letters in $S$:}

\begin{itemize}
\item For each unmatched letter $0$ in $S$, we add three pairs $(a,b'),(e,b'),(a,e)$ to $A_S$ in  its corresponding substring $S_0 = (a(eb')(ae)b')$. Similarly, for each unmatched letter $1$, we  add three pairs to $A_S$.
\item For each unmatched letter $0'$ in $S$, we add two pairs $(b,b),(a',a')$ to $A_S$ in  its corresponding substring $S_0 = (bb)(a'a')$. Similarly, for each unmatched letter $1'$, we  add two pairs to $A_S$.
\end{itemize}

The set $A_M$ has size $2\RNA(S)$, the set $A_S$ has size $\frac{|S_{\text{Dyck}}| - 4\RNA(S)}{2}$, and $D$ is an empty set. Therefore, the cost of $A$ is $\frac{|S_{\text{Dyck}}| - 4\RNA(S)}{2} = \frac{|S_{\text{Dyck}}|}{2} - 2\RNA(S)$.

\paragraph{Lower Bound.} We show that the Dyck edit distance of $S_{\text{Dyck}}$ is at least $\frac{|S_{\text{Dyck}}|}{2} - 2\RNA(S)$. Given a crossing-free matching $A$ (on the string $S_{\text{Dyck}}$) of cost $C$, we  recover an RNA folding of $S$ that has $\geq \frac{|S_{\text{Dyck}}|}{4} - \frac{C}{2}$ number of matched pairs.

We build a multi-graph $G=(V,E)$ such that $V$ is the set  of all substrings $S_0$, $S_{0'}$, $S_1$, and $S_{1'}$ that constitute $S_{\text{Dyck}}$, and the number of  edges between two substrings in $V$ is the number of pairs in $A_M$ linking letters between these two substrings. Note that $|V|=n$ and $|E|=A_M$. It is clear that $C \geq \frac{|S_{\text{Dyck}}| - 2|E|}{2}$, since $|A_S|+|D| \geq \frac{|S_{\text{Dyck}}| - 2|A_M|}{2} = \frac{|S_{\text{Dyck}}| - 2|E|}{2}$. We show that we can obtain an RNA folding of $S$ that has size $\geq \frac{|E|}{2}$. Note that $\frac{|E|}{2} \geq \frac{|S_{\text{Dyck}}|}{4} - \frac{C}{2}$.
We make the following three observations.
\begin{description}
\item[(O1)] $G$ has degree at most 2. The reason is that at most two letters in each substring  $S_0$, $S_{0'}$, $S_1$, $S_{1'}$ can participate in pairings of the form $(x,x')$,  $x \in \{a,b,c,d\}$, without crossing.
\item[(O2)] In the graph $G$, each edge either (i) links a substring $S_0$ with a substring $S_{0'}$,  or (ii) links a substring $S_1$ with a substring $S_{1'}$.  The reason is that any pairing of the form $(x,x')$,  $x \in \{a,b,c,d\}$, must be made between $S_0$ and $S_{0'}$, or between $S_1$ and $S_{1'}$.
\item[(O3)] $G$ does not contain any cycle of odd length. This is due to (O2).
\end{description}
In view of (O2), a (graph-theoretic) matching $M \subseteq E$ of $G$ naturally corresponds to a size-$|M|$ RNA folding of $S$, as follows. For each edge, which is a pair of substrings in $S_{\text{Dyck}}$, in $M$, we add its corresponding pair of letters in $S$ to the RNA folding. By (O1) and (O3), $G$  in a graph of maximum degree 2 without odd cycles, and a maximum matching in such a graph has size at least $\frac{|E|}{2}$, and so we conclude the proof.
\end{proof}

The reason that the letter $e$ is essential in the proof is briefly explained as follows. Suppose that $e$ is removed. For  each matched pair $(0,0')$ in the RNA folding of $S$,  after adding two pairs $(a,a')$ and $(a,a')$  to $A_M$, the letter $b'$ between two letters $a$ in $S_0=ab'ab'$ cannot participate in any matching. Hence  some letters have to be in $D$ according to our construction of the crossing-free matching $A$, which implies that our construction might not be optimal.

Consider the case $S = (0 0' 0')$. We would have $S_{\text{Dyck}} = ab'ab'bba'a'bba'a'$ after removing $e$. If we match the two pairs $(a,a')$ and $(a,a')$ in $\mathbf{a}b'\mathbf{a}b'bb\mathbf{a'a'}bba'a'$, then the cost will be at least $5$ (three substitutions and two deletions are needed). However, there is a solution that uses only 4 substitutions: $\mathbf{a}(b'\mathbf{a}(b'(bb)a')\mathbf{a'}(bb)a')\mathbf{a'}$.

Note that if substitution is not allowed in the definition of Dyck edit distance, then the letter $e$ in the above proof is not needed, and this lowers the alphabet size requirement from 10 to 8.

\section{Conclusion}
In this paper we present a conditional lower bound of RNA folding problem with alphabet size 4, and demonstrate a simple reduction from RNA folding problem to Dyck edit distance problem. One open problem that still remains is whether it is possible to reduce Dyck edit distance problem to RNA folding problem (i.e., the reverse of Theorem~\ref{thm-3}).
The ``standard'' RNA folding problem only finds an optimal pseudoknot-free fold for an RNA sequence;  however, the ``real world'' RNA folding includes  pseudoknots, and is more complicated. There are variants of RNA folding problem that consider pseudoknots; see e.g.,~\cite{FG12} and the citations therein. It would be interesting to see whether the techniques presented in this paper and~\cite{ABV15,BK15} can be adapted to provide meaningful lower bounds for these problems.

\bigskip
\noindent{\bf Acknowledgements.} The author thanks Seth Pettie for helpful discussions and comments.

\bibliographystyle{plain}
\bibliography{references}

\end{document}